\documentclass[12pt]{article}
\usepackage{amsmath,amssymb,fullpage,graphicx}
\usepackage[round]{natbib}

\newtheorem{theorem}{Theorem}
\newtheorem{lemma}[theorem]{Lemma}

\newenvironment{proof}{{\bf Proof.}}{$\Box$}

\renewcommand{\Re}{\mbox{$\mathbb{R}$}}

\newcommand\Hess{{\cal H}}
\newcommand\dest{{\sf dest}}
\newcommand\Ascend{{\cal A}}
\newcommand\Modes{{\cal M}}
\newcommand\half{{\frac12}}

\newenvironment{enum}{
\begin{enumerate}
  \setlength{\itemsep}{1pt}
  \setlength{\parskip}{0pt}
  \setlength{\parsep}{0pt}
}{\end{enumerate}}

\usepackage{fouriernc}

\let\hat\widehat
\let\tilde\widetilde

\begin{document}

\begin{center}
{\bf\Large Nonparametric Inference For Density Modes}\\
\vspace{.5cm}
{\bf Christopher R. Genovese, Marco Perone-Pacifico,\\
Isabella Verdinelli and Larry Wasserman}\\
Carnegie Mellon University and University of Rome
\end{center}

\begin{center}
December 29, 2013
\end{center}

\begin{quote}
We derive nonparametric confidence intervals
for the eigenvalues of the Hessian at modes
of a density estimate.
This provides information about the strength
and shape of modes and
can also be used
as a significance test.
We use a data-splitting approach
in which potential modes are identified using the
first half of the data and
inference is done with the second half of the data.
To get valid confidence sets for the eigenvalues,
we use
a bootstrap based on 
an elementary-symmetric-polynomial (ESP) transformation.
This leads to valid bootstrap confidence sets
regardless of any multiplicities in the eigenvalues.
We also suggest a new method for bandwidth selection, namely,
choosing the bandwidth to maximize the number of significant modes.
We show by example that this method works well.
Even when the true distribution is singular,
and hence does not have a density,
(in which case cross validation chooses a zero bandwidth),
our method chooses a reasonable bandwidth.
\end{quote}


\baselineskip=16pt

\section{Introduction}

Figure \ref{fig::motivation}
shows a one-dimensional density estimate with
two modes.
The leftmost mode is likely to correspond to a real mode
in the true density.
But the second smaller mode on the right may be due to
random fluctuation.
How can we tell a real mode from random fluctuation?
In this paper, we provide a simple hypothesis test to answer this question
that is easy to implement, even in multivariate problems.
The basic idea is this:
a confidence interval for the second derivative of the density
will be strictly negative for the left mode
but is likely to cross 0 for the right mode.

Let $Z_1,\ldots, Z_n\in \mathbb{R}^d$ be a sample from a 
distribution $P$ with density $p$.
We assume that the gradient $g$ and Hessian $\Hess$ of $p$
are bounded continuous functions.
Furthermore,
we assume that $p$ has finitely many, well-separated modes
$m_1,\ldots, m_{k_0}$.
We do not assume that $k_0$ is known.
Our goal is to estimate the modes
and to give confidence sets that provide
shape information about the estimated modes.

There are many reasons for mode hunting
and many methods to find modes;
see, for example,
\cite{klemela2009smoothing, li2007nonparametric, 
dumbgen2008multiscale}.
In particular, modes can be used as the basis of
nonparametric clustering \citep{chacon, chazal2011persistence, Comaniciu,
Fukunaga,li2007nonparametric}.

There are several difficulties
in defining tests for modes.
Consider a point $x\in \mathbb{R}^d$
and suppose we want to test
$$
H_0: \ x\ \mbox{is not a mode\ of }p\ \ \ \ \ \ \ \ \mbox{versus}\ \ \ \ \ \ \ \ 
H_1: \ x\ \mbox{is a mode\ of }p.
$$
First, testing the null hypothesis of ``no mode'' 
raises problems, analogous to testing the null that a mean is not zero,
because the alternative forms a measure zero set.
More precisely, 
if $\nabla p(x) \equiv g(x) = 
(g_1(x),\ldots,g_d(x))^T$ is the gradient of $p$ at $x$,
$\lambda_1(x) \geq \cdots \geq \lambda_d(x)$ are the 
eigenvalues of the Hessian $\Hess(x)$,
and $\Omega = \mathbb{R} \times \mathbb{R}^d$,
then 
$H_0 = \Omega - H_1$ and
$$
H_1=\Bigl\{ (\lambda_1,g)\in \Omega:\ \lambda_1 < 0,\ 
g = (0,\ldots, 0)^T \Bigr\}.
$$
is a measure zero subset of $\Omega$.
No meaningful test can be constructed
of such a ``reverse null hypothesis.''
The second problem is that there are uncountably
many possible locations at which a mode can occur,
leading potentially to a difficult multiple testing problem.
Finally, verifying that a mode exists requires making
inference about eigenvalues of the Hessian.
But the eigenvalues are not continuously differentiable 
functions of the Hessian
which makes methods like the bootstrap and the delta method
invalid.

\begin{figure}
\begin{center}
\includegraphics[scale=.3]{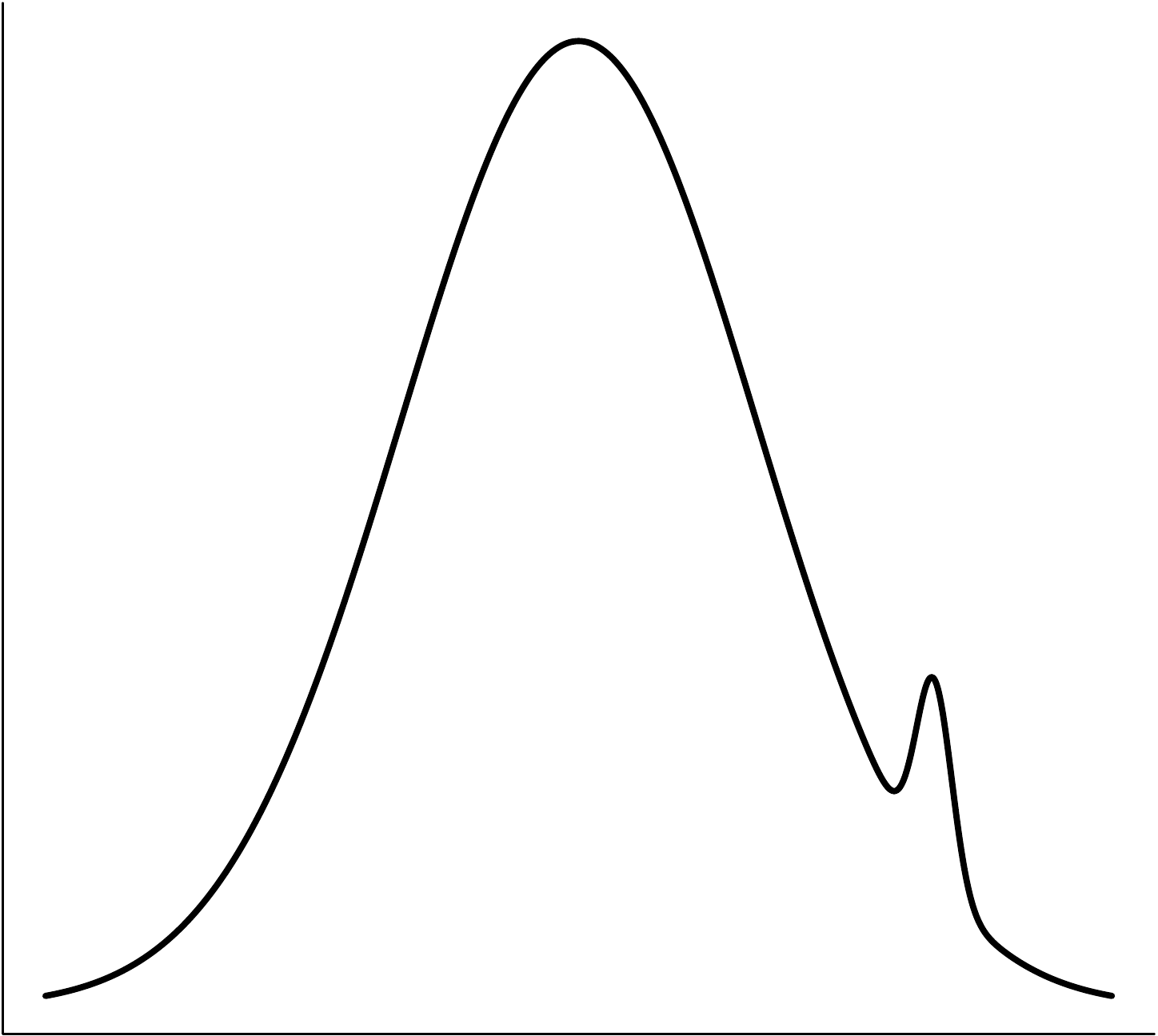}
\end{center}
\caption{\em The mode on the left appears to be real.
The mode on the right might be due to random fluctuation.}
\label{fig::motivation}
\end{figure}

We overcome these problems by 
combining several ideas:
\begin{enum}
\item We use data splitting to separate the process of finding
candidate modes from the process of hypothesis testing.
This ameliorates the multiplicity problem
and simplifies the hypothesis test as well.
Specifically, assume that the sample size is $2n$
and randomly split the data into two halves
$X =(X_1,\ldots, X_n)$ and
$Y =(Y_1,\ldots, Y_n)$.

\item In stage one, we use $X$ to find
a finite set of candidate modes
$\hat{\cal M}$.

\item 
In stage two,
we use the second half of the data $Y$
to estimate the Hessian of the density
at the candidate modes
$\hat{\cal M}$.
We transform the eigenvalues of the Hessian using 
elementary symmetric polynomials (ESP).
As noted in \cite{beran1985bootstrap},
the bootstrap leads to
asymptotically valid confidence sets for the
transformed eigenvalues.
We then invert the mapping to get a valid
confidence set for the eigenvalues.
This provides useful shape information about
the modes, which we call an eigenportrait.

\item 
The eigenportrait can be used to
formulate a test for the importance of the mode.
As a surrogate for testing
whether a candidate mode is not really a mode,
we instead test if $x$ is an ``approximate mode''.
This requires reformulating $H_1$
to capture the idea of an approximate mode.
There is no unique way to do this.
One possibility is to take
$H_0 = \Omega - H_1$ where
$$
H_1=\Bigl\{ (\lambda_1,g)\in \Omega:\ \lambda_1 < 0,\ 
||g||<\delta \Bigr\}
$$
where $\delta>0$ is a small positive constant.
In practice, the constraint
$||g||<\delta$ has no effect on the test
since the estimated gradient is 0 at the modes in stage one
and hence is likely to be close to 0 in stage two.
In practice, therefore, we simplify matters by just testing
$$
H_0: \lambda_1 \geq 0\ \ \ {\rm versus}\ \ \ H_1: \lambda_1 < 0.
$$
\end{enum}

{\bf Bias.}
We will use a kernel density estimator
$\hat p_h$ depending on a bandwidth $h>0$.
In this paper we view
$\hat p_h$ as an estimator of its mean $p_h$.
In particular, we view the modes
of $\hat p_h$
as estimates of the modes of $p_h$.
Of course, there is a bias
(typically of order $O(h^2)$)
that separates $p_h$ from $p$.
This bias is not of critical importance
when studying modes.
Instead, our primary concern is the variability
of $\hat p_h$ as an estimator of $p_h$.
Including the bias in any inferential procedures
for density estimators
raises well known complications since
the bias is harder to estimate than the density.
One can use various devices such as undersmoothing to
deal with the bias.
These difficulties are a distraction from our main thrust
and so we focus on inference for $p_h$.

\vspace{1cm}

{\bf Related Work.}
There is a large literature on mode finding.
Many methods are based on the mean-shift algorithm for finding modes
of kernel estimators; see \cite{Comaniciu, Fukunaga}.
An early paper in the statistics literature
on using kernel density estimators for mode hunting is
\cite{silverman1981using}.
He used the observed bandwidth at which a new mode appears
as a test for multimodality.
The properties of this test
are rather complicated, even in one-dimension: see
\cite{mammen1992some}.

Significance testing for modes of kernel estimators
was considered in
\cite{godtliebsen2002significance} and
\cite{duong2008feature}.
The latter reference is very related to this paper.
We discuss the differences in our approaches in Section 
\ref{section::testing}.
Asymptotic theory and bandwidth selection for
mode hunting and derivative estimation is discussed in
\cite{ChaconAndDuong, chacon2010multivariate, chacon2011asymptotics}.
\cite{donoho1991geometrizing} showed that the minimax rate
for estimating a mode in one dimension, assuming the density
is locally quadratic around the mode,
is $O(n^{-1/5})$.
Although not stated explicitly in that paper,
it is clear that the rate for $d$-dimensional densities
is $O(n^{-1/(4+d)})$.
\cite{konakov1974asymptotic}
studied the asymptotics of the mode estimator
in the multivariate case.
\cite{klemela2005adaptive}
considered adaptive estimation that takes into account the regularity
in a neighborhood of a mode.
\cite{dumbgen2008multiscale}
presented a method for constructing
multiscale confidence intervals for modes
but the method is only applicable to one-dimensional densities.

Clustering, based on modes,
was used in 
\cite{chacon} and
\cite{li2007nonparametric}.
\cite{chazal2011persistence}
considered a completely different approach to
mode-based clustering
on persistent homology; we compare this to the current approach
in Section \ref{section::persistence}.
Finally, we mention that there is a large literature on the related
problem of estimating level sets of density;
for example, see \cite{polonik1995measuring, cadre2006kernel, 
walther1997granulometric}.
The concept of
excess mass \cite{muller1991excess} 
provides a link between level sets and modes.

\vspace{1cm}

{\bf Outline.}
In Section \ref{section::modesandclusters},
we discuss mode hunting and mode clustering.
We present our hypothesis test in Section
\ref{section::testing}.
A crucial part of the test is a non-standard
bootstrap procedure described in
Section \ref{section::esp}.
We compare our approach to persistent homology in
Section \ref{section::persistence}.
Section \ref{sec::examples} presents some examples.
In Section \ref{section::bandwidth},
we use our procedure as part of a
new method for bandwidth selection for mode hunting.
Section \ref{section::theory} presents some
theoretical properties of the method.
Concluding remarks are in Section \ref{section::discussion}.

{\bf Notation.}
Given a density function $p$,
we use $g(x)$ to denote the gradient of $p$ at $x$
and we use $\Hess(x)$
to denote the
Hessian of $p$ at $x$.
The eigenvalues of $\Hess(x)$
are denoted by
$\lambda(x) = (\lambda_1(x),\ldots,\lambda_d(x))$ where
$\lambda_1(x) \geq \cdots \geq \lambda_d(x)$.
Since the eigenvalues at a mode are negative,
it is convenient to define
$\gamma(x) = (\gamma_1(x),\ldots,\gamma_d(x))$ where
$\gamma_j(x) = -\lambda_j(x)$.
For an $n\times r$ matrix $A$, define ${\rm vec}(A)$
to be the $nr \times 1$ column vector obtained by
stacking the columns of $A$, that is,
${\rm vec} A =(A_{11}, A_{21},\ldots,A_{n1},A_{12},\ldots, A_{nn})^T$.
Also, for symmetric matrices,
${\rm vech}$ is the vec operator applied only to the
upper triangular part of the matrix.
For a vector-valued function
$f=(f_1,\ldots, f_d)$
we follow \cite{chacon2011asymptotics},
by defining $D^{\otimes r} f$ as
$$
D^{\otimes r}f(x) = \left( \begin{array}{c}
                           D^{\otimes r}f_1(x)\\ 
                           \vdots \\
                           D^{\otimes r}f_d(x)
                          \end{array} \right).
$$
Here, $D^{\otimes r}$ denotes the $r^{\rm th}$ derivative.
Then, for the Hessian $\Hess f = \partial^2 f/(\partial x \partial x^T)$ we have
${\rm vec}\,\Hess f = D^{\otimes 2}f$.
In the special case $r=1$ we usually just write
$\nabla f$ for the gradient.
Also, we sometimes
use $\nabla^{(2)}$ for the second derivative.
The largest eigenvalue of a matrix $A$ is denoted by
$\lambda_1(A)$.
We use $C$ to denote a generic positive constant.

{\bf Assumptions.}
Throughout the paper we make the following assumptions.

(A1) The density $p$ is a bounded, continuous
density supported on a compact set
${\cal X}\subset \mathbb{R}^d$.

(A2) The gradient $g$ and Hessian ${\cal H}$ of $p$
are bounded and continuous. The Hessian is non-degenerate at all stationary points.

(A3) $p$ has finitely many 
modes $m_1,\ldots, m_{k_0}$
in the interior of ${\cal X}$.

(A4)
Let
\begin{equation}
\Delta = \min_{s\neq t}||m_s - m_t|| \ \ \ \ \ \mbox{and}\ \ \ \ \ 
L=\max_{1\leq j\leq k_0}\lambda_1({\cal H}(m_j)).
\end{equation}
We assume that 
$\Delta >0$ and $L < 0$.

(A5) The kernel $K$ used in the density
estimator is a symmetric probability density with bounded and continuous first
and second derivatives and bounded second moment.

\section{Modes and Clusters}
\label{section::modesandclusters}

One of our main motivations for finding significant modes
is so that they can be used for clustering.
Let $m_1,\ldots, m_{k_0}$ be the modes of $p$.
Assume that $p$ is a \emph{Morse function},
which means that the Hessian of $p$ at each
stationary point is non-degenerate.

Given any point $x\in \mathbb{R}^d$
there is a unique gradient ascent path, or integral curve,
passing through $x$
that eventually leads to one of the modes.
We define the clusters to be the ``basins of attraction'' of the modes,
the equivalence classes of points whose ascent paths lead to the same mode.
Formally, an integral curve through $x$ is a path
$\pi_x: \mathbb{R}\to \mathbb{R}^d$
such that $\pi_x(t)=x$ for some $t$ and such that
\begin{equation}\label{eq::diffeq}
\pi'_x(t) = \nabla p(\pi_x(t)).
\end{equation}
Integral curves never intersect (except at stationary points)
and they partition the space
(\cite{Matsumoto}).
Equation (\ref{eq::diffeq})
means that the path $\pi$
follows the direction of steepest ascent of $p$ through $x$.
The destination of the integral curve $\pi$ through a (non-mode) point $x$ is defined by
\begin{equation}
\dest(x) = \lim_{t\to\infty}\pi_x(t).
\end{equation}
(We define $\dest(m) = m$ for any mode $m$.)
It can then be shown that for all $x$,
$\dest(x) = m_j$ for some mode $m_j$.
That is: all integral curves lead to modes.
For each mode $m_j$, define the sets
\begin{equation}\label{eq::ascending}
\Ascend_j = \Bigl\{ x:\ \dest(x) = m_j \Bigr\}.
\end{equation}
These sets are known as the \emph{ascending manifolds},
and also known as
the cluster associated with $m_j$,
or the basin of attraction of $m_j$.
The $\Ascend_j$'s partition the space.
See Figure \ref{fig::Morse}.

\begin{figure}
\begin{center}
\includegraphics[scale=.8]{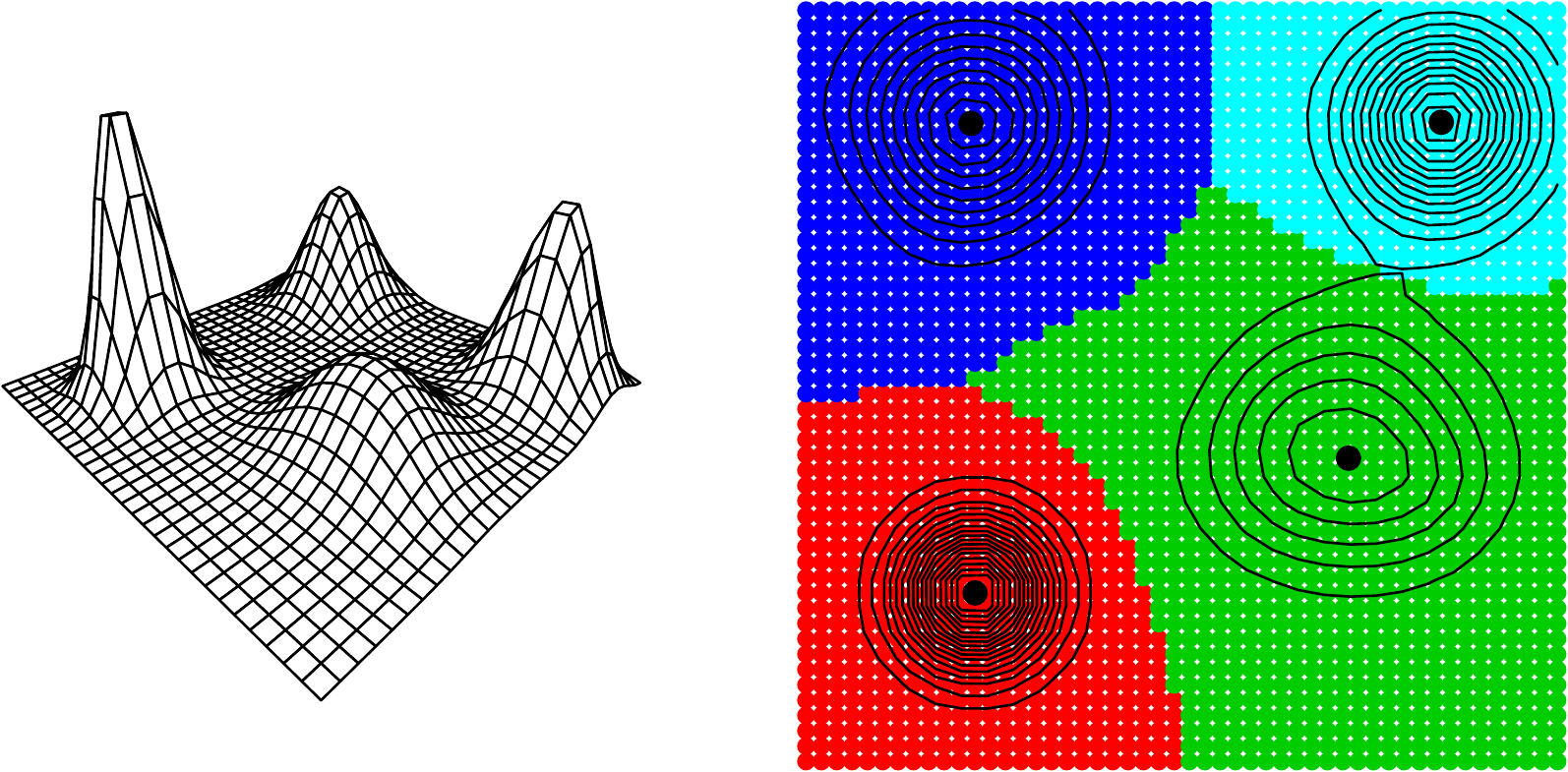}
\end{center}
\caption{\em The left plot shows a function with four modes.
The right plot shows the ascending manifolds (basins of attraction) corresponding to the four modes.}
\label{fig::Morse}
\end{figure}

Given data
$X_1,\ldots, X_n$
we construct an estimate $\hat p$ of the density.
Let $\hat m_1,\ldots, \hat m_k$
be the estimated modes
and let
$\hat\Ascend_1,\ldots, \hat\Ascend_k$
be the corresponding ascending manifolds derived from $\hat p$.
The sample clusters $C_1,\ldots, C_k$ are defined to be
$C_j = \bigl\{X_i:\ X_i\in \hat\Ascend_j\bigr\}$.
Before finding clusters,
it is important to find out which modes
are significant and which are explainable
as random fluctuations.
This is one of the motivations for the current paper.

We will estimate the density $p$ 
with the kernel density estimator 
\begin{equation}
\hat p(x) \equiv \hat p_h(x) =
\frac{1}{n}\sum_{i=1}^n \frac{1}{h^d}\,K\left(\frac{||x-X_i||}{h}\right)
\end{equation}
where $K$ is a smooth, symmetric kernel and $h>0$
is the bandwidth.
The mean of the estimator is
\begin{equation}
p_h(x) = \mathbb{E}[\hat p_h(x)] =
\int K(t) p(x + th) dt.
\end{equation}
In general, one can use a bandwidth matrix $H$
in the estimator, with 
\begin{equation}
\hat p(x) \equiv \hat p_H(x) = \frac{1}{n}\sum_{i=1}^n K_H(x-X_i)
\end{equation}
where $K_H(x) = |H|^{-\frac{1}{2}}K(H^{-\frac{1}{2}}x)$.
As discussed in \cite{ChaconAndDuong} and \cite{chacon2011asymptotics},
using a non-diagonal matrix can lead to
better density estimates than using a diagonal bandwidth matrix.
But for simplicity,
here we use a single, scalar bandwidth $h$,
corresponding to $H = h^2 I$.
As explained in the introduction,
in this paper we regard $\hat p_h$
as an estimator of $p_h$
and we aim to find the modes of $p_h$.

To locate the modes of $\hat p_h$
we use the \emph{mean shift algorithm} (\citep{Fukunaga, Comaniciu}),
which finds modes by approximating the steepest ascent paths.
(\cite{castromason}).
The algorithm is given in Figure \ref{fig::algorithm}.
The result of this process is a set of candidate modes
$\hat{\cal M} = \{\hat m_1,\ldots, \hat m_k\}$.
Note that $k$ is random
since it is the observed number of modes of the density estimator.

\begin{figure}
\fbox{\parbox{6in}{
\vspace{.3cm}
\begin{center}
\sf Mean Shift Algorithm\\
\end{center}
\begin{enum}
\item 
Input: $\hat p(x)$ and a mesh of points 
$A = \{a_1,\ldots, a_N\}$
(often taken to be the data points).
\item 
For each mesh point $a_j$,
set $a_j^{(0)} = a_j$ and
iterate the following equation until convergence:
$$
a_j^{(s+1)} \longleftarrow 
\frac{\sum_{i=1}^n X_i K\left(\frac{||a_j^{(s)}-X_i||}{h}\right)}
{\sum_{i=1}^n K\left(\frac{||a_j^{(s)}  -X_i||}{h}\right)}.
$$
\item Let
$\hat\Modes$ be the unique values of the set
$\{a_1^{(\infty)},\ldots,a_N^{(\infty)}\}$.
\item Output: $\hat\Modes$.
\end{enum}
}}
\caption{\em The Mean Shift Algorithm. (\cite{Fukunaga}; 
\cite{Comaniciu})}
\label{fig::algorithm}
\end{figure}

\section{The Method}
\label{section::testing}

For simplicity, assume that the sample size is even
and let $2n$ denote the sample size.
Our testing procedure involves the following steps:
\begin{enum}
\item Split the data randomly into two halves
$X=(X_1,\ldots, X_n)$ and
$Y=(Y_1,\ldots, Y_n)$, say.
\item Use $X$ to construct a density estimate 
$\hat p_{X,h}$ and find
candidate modes $\hat m_1,\ldots, \hat m_k$.
\item Use
$Y$
to construct another density estimate $\hat p_{Y,h}$
and compute
the Hessian
$\hat{\cal H}_{Y,h}$ of $\hat p_{Y,h}$
at each $\hat m_j$, where $1\leq j \leq k$.
Let $\hat\lambda_j = (\hat\lambda_{1j},\ldots, \hat\lambda_{dj})$
be the eigenvalues of
$\hat {\cal H}_{Y,h}(\hat m_j)$ and let
$\hat\gamma_j = (-\hat\lambda_{1j},\ldots, -\hat\lambda_{dj})$.
\item Construct a $1-\alpha/k$ confidence rectangle
$G_j$ for $\gamma_j = (\gamma_{1j},\ldots, \gamma_{dj})^T$
where
$\gamma_{sj} = - \lambda_{s}(\Hess_h(\hat m_j))$.
The collection of confidence rectangles
$G_1,\ldots,G_k$ is called the eigenportrait.
From $G_j$
we get a confidence interval
${\cal C}_j$ for the the leading
eigenvalue $\gamma_{1j} = - \lambda_{1}({\cal H}_{Y,h}(\hat m_j))$.
\item Reject
$H_0: \gamma_{1j} <0$ if\ \ $\inf\,\bigl\{x\in {\cal C}_j\bigr\} > 0$
and declare $\hat m_j$ to be a real mode.
\end{enum}

There are $k$ candidate modes.
At each mode,
we have $d$-dimensional vectors
$\lambda_j=(\lambda_{1j},\ldots,\lambda_{dj})$
and
$\gamma_j=(\gamma_{1j},\ldots,\gamma_{dj})$.
Here are some remarks on the steps.

{\bf Step 1 and 2:}
The purpose of the data splitting
is to assure the validity of the confidence intervals.
If we did not split the data,
we could instead get a valid test
by treating the estimated Hessian as a 
stochastic process over the whole space
and then estimating the maximum fluctuations of this process.
While this is possible,
splitting the data and focusing on finitely many points is much simpler.

{\bf Step 3:}
We estimate the Hessian at $\hat m_j$, $\hat{\cal H}_{Y,h}(\hat m_j)$, by
using the Hessian of the
density estimator from the second half of the data.
Specifically, with $H = h^2 I$, 
\begin{equation}
{\rm vec}\,\hat{\cal H}_{Y,h}(\hat m_j) = 
\frac{1}{n} |H|^{-\half} 
(H^{-\half})^{\otimes 2} \sum_{i=1}^n D^{\otimes 2}K(H^{-\half}(\hat m_j -Y_i)).
\end{equation}

{\bf Step 4.}
Using the method described later in Section \ref{section::esp},
we construct $1-\alpha/k$ confidence intervals
${\cal C}_j$ for
$\gamma_{1j}$, $j=1,\ldots, k$.
The validity of the bootstrap in
Section \ref{section::esp}, together with the independence
from sample splitting, ensures that
$$
\liminf_{n\to\infty}\ \mathbb{P}(\gamma_{j} \in G_j,\ {\rm for\ all\ }j)\geq 1-\alpha.
$$
We test
$$
H_{0j}: \gamma_{1j} \leq 0\ \ \ \ \ \mbox{versus}\ \ \ \ \ H_{1j}: \gamma_{1j} >0
$$
for $j=1,\ldots, k$ and we reject
$H_{0j}$ if the confidence set ${\cal C}_j$ lies above 0.

{\bf Step 5.}
In principle, we would like to test the null hypothesis 
{\em $H_{0j}: \ \hat m_j$ is not a mode} versus the alternative
{\em $H_{1j}: \ \hat m_j$ is a mode}
for $j=1,\ldots, k$.
But, as we explained earlier
it is not possible to construct a non-trivial test for
this hypothesis since $H_{1j}$ has measure 0.
Instead we could replace $H_{1j}$
with the statement: ``$\hat m_j$ is an approximate mode''.
This suggests testing
$\tilde H_{0j}$ versus $\tilde H_{1j}$ 
where
$$
\tilde H_{1j} = \Bigl\{ (\lambda_1,g):\ -\lambda_1 <0,\ \ ||g|| \leq \delta \Bigr\}
$$
for some $\delta>0$,
and
$\tilde H_{0j} = \tilde H_{1j}^c$.
However, thanks to the data-splitting,
testing $\tilde H_{0j}$ versus $\tilde H_{1j}$ 
is asymptotically equivalent to
testing $H_{0j}$ versus $H_{1j}$.
This follows since
$$
||\hat g_{Y,h}(\hat m_j)|| = ||\hat g_{X,h}(\hat m_j)|| + O_P \left(\frac{1}{n h^{d+2}}\right) =
0 + O_P \left(\frac{1}{n h^{d+2}}\right) =  o_P(1).
$$
Hence, with probability tending to 1,
$|| \hat g_{Y,h}(\hat m_j)|| < \delta$
and, asymptotically,
we reject
$\tilde H_{0j}$ if and only if
we reject
$H_{0j}$.
In summary,
we interpret the rejection
of $H_0: \gamma_{1j}<0$
to mean that $\hat m_j$ is an approximate mode.

{\bf Comparison with \cite{duong2008feature}.}
\cite{duong2008feature} describe an approach with several features similar to ours.
They carry out two statistical tests:
that the gradient is 0 and that the norm of the Hessian is 0.
They test these hypotheses at a large number of points,
with a multiple testing correction.
Regions where the gradient null is not rejected
and the Hessian null is rejected are deemed interesting.
Plotting these regions provides a useful visualization of the density's behavior.
Note that the hypotheses used and the goals are quite different between the
two methods.
Their method is more exploratory and provides effective visualizations.
Our method is intended to produce a definite, finite set of potential
modes, with a test for the significance of each.
Further, our goal is to provide a set of confidence intervals for
the eigenvalues of the Hessian at the estimated modes,
as we describe in the next section.

\section{The Telepathic Bootstrap}
\label{section::esp}

To implement the test described in the previous section,
we need to
construct a confidence interval for
$\gamma_1(x) = -\lambda_1(x)$,
for $x\in \hat\Modes$,
which requires some care.
Let 
\begin{equation}
\hat\lambda_1(x) \geq \hat\lambda_2(x) \geq \cdots \geq \hat\lambda_d(x)
\end{equation}
denote the eigenvalues of $\hat{\cal H}_{Y,h}(x)$. 
We construct confidence regions for the eigenvalues 
using the bootstrap.
Bootstrapping the eigenvalues poses some problems.
In general,
$\lambda(x) = (\lambda_1(x),\ldots, \lambda_d(x))$
is not a continuously differentiable function of ${\cal H}_h(x)$,
the Hessian of $p_h$.
As a result,
standard bootstrapping applied to the Hessian will not
produce valid confidence sets for the eigenvalues.
However, Beran and Srivastava (1985) note 
that if the eigenvalues are transformed  using elementary symmetric 
polynomials, then the confidence set obtained is valid
as we now explain.

Given ordered, not necessarily distinct, eigenvalues
$\lambda_1(x) \ge \lambda_2(x) \ge  \cdots \ge \lambda_d(x)$,
define the elementary symmetric polynomials (ESP) by
\begin{align} \label{eqn::one-to-one}
s_1(x)    &= \sum_{i=1}^d \lambda_i(x) \nonumber \\
s_2(x)    &= \sum_{i_1=1}^d \; \sum_{i_2=i_1+1}^d \lambda_{i_1}(x) \cdot \lambda_{i_2}(x)\nonumber \\
\cdots & \cdots   \\
s_k(x)    &= \sum_{i_1=1}^d \: \sum_{i_2=i_1+1}^d \ldots 
         \sum_{i_k=i_{k-1}+1}^d \lambda_{i_1}(x) \cdot \lambda_{i_2}(x) \ldots \cdot \lambda_{i_k}(x)
         \nonumber \\
\cdots & \cdots \nonumber \\
s_d(x)    &= \lambda_1(x) \cdot \lambda_2(x) \cdot \ldots \cdot \lambda_d(x). \nonumber
\end{align}
Conversely, $\lambda_1(x), \cdots, \lambda_d(x)$ are roots of the characteristic
polynomial
\begin{equation}
P(\lambda(x)) \ = \ \prod_{i=1}^d (\lambda_i(x) - \lambda(x))\ = 
\  (-1)^d\lambda^d(x) + \sum_{k=1}^d (-1)^k \ s_k  \ \lambda^{d-k}(x) = 0.
\end{equation}
Let $s(x)=(s_1(x),\ldots, s_d(x))$.
Note that all the eigenvalues are negative
if and only if
$(-1)^k s_k >0$ for all $k$.
Also, $s(x)$ is a continuously differentiable function of ${\cal H}_h(x)$ and
the map from $\lambda(x)$ to $s(x)$ is one-to-one.
Hence, we can write $s(x) = w(\lambda(x))$ and $\lambda(x) = w^{-1}(s(x))$.
See Figure \ref{fig::ESP}.

\begin{figure}
\begin{center}
\includegraphics[scale=.7]{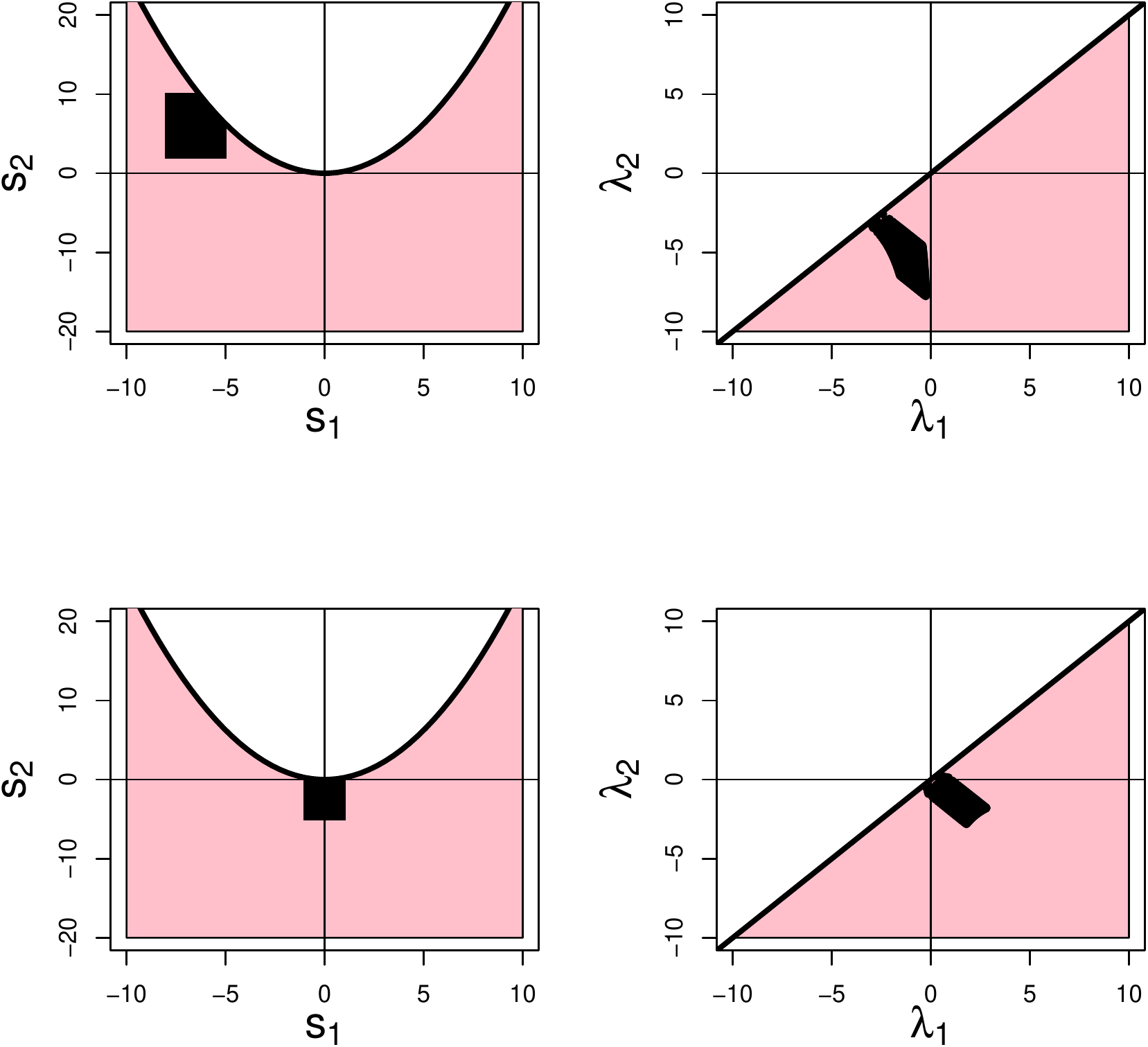}
\end{center}
\caption{\em 
Left: Examples of rectangles ${\cal S}$ in terms of $(s_1,s_2)$.
Right: Corresponding sets of eigenvalues $w^{-1}({\cal S})$.}
\label{fig::ESP}
\end{figure}

The steps in the bootstrap, at a particular candidate mode $\hat m_j$ 
(see Figure \ref{fig::ModeTest})
are as follows
(we suppress the subscript $j$): 
\begin{enum}
\item Let $\hat\lambda$ be the eigenvalues of the estimated Hessian and let $\hat s = w(\hat\lambda)$.
\item Draw $Y_1^*,\ldots, Y_n^* \sim P_n$ where $P_n$ is the empirical distribution
of $Y_1,\ldots, Y_n$.
\item Compute the density estimate, the Hessian and the estimates eigenvalues
$\lambda^*=(\lambda_1^*,\ldots, \lambda_d^*)$.
Compute the ESP-transformed eigenvalues
$s^*=(s_1^*,\ldots,s_d^*) = w(\lambda^*)$.
\item Repeat steps 2 and 3 $B$ times
yielding $B$ vectors
$s^{*1},\ldots,s^{*B}$.
\item Find the $1-\alpha/k$ bootstrap quantile $q$ defined by:
$$
\frac{1}{B}\sum_{b=1}^B I\Bigl( ||s^{*b} - \hat s||_\infty > q\Bigr) = \frac{\alpha}{k}.
$$
The set
\begin{equation}\label{eq::setS}
{\cal S} = \Bigl\{ s:\ ||s-\hat s||_\infty \leq q\Bigr\}
\end{equation}
is a $1-\alpha$ asymptotic confidence set for
$s=(s_1,\ldots, s_d)$.
\item Let
\begin{equation}
{\cal C} = \Biggl[ \min_{b\in J} (-\lambda^{*b}_1),\ \ \max_{b\in J} (-\lambda^{*b}_1)\Biggr]
\end{equation}
where $J = \Bigl\{1 \le b \le B:\ s^{*b}\in {\cal S}\Bigr\}.$
\end{enum}

The above procedure is used at each 
candidate mode $\hat m_j$
and hence we get
confidence sets
${\cal S}_1,\ldots, {\cal S}_k$
for $s_1,\ldots,s_k$,
confidence rectangles
$G_1,\ldots, G_k$ for the $\gamma_j$'s and
confidence intervals
${\cal C}_1,\ldots, {\cal C}_k$ for
$\gamma_{11},\ldots, \gamma_{1k}$.
Here,
$s_j = (s_{1j},\ldots, s_{dj})$
and $\gamma_{1j}$ is minus the largest eigenvalue of the Hessian at mode $\hat m_j$.

The last two steps deserve some explanation.
A confidence set
for $\gamma_{1j}$ at $\hat m_j$
is $w^{-1}({\cal S}_j)$.
From Corollary 1
of \cite{beran1985bootstrap}, 
it follows that
$$
\liminf_{n\to\infty}\ \mathbb{P}(s_j\in {\cal S}_j) \geq 1-\frac{\alpha}{k}
$$
and hence
$$
\liminf_{n\to\infty}\ \mathbb{P}(\gamma_{1j}\in w^{-1}({\cal S}_j)) \geq 1-\frac{\alpha}{k}
$$
Therefore,
\begin{equation}
\liminf_{n\to\infty}\ \mathbb{P}(\gamma_{1j}\in w^{-1}({\cal S}_j)\ {\rm for\ each\ }j) \geq 1-\alpha.
\end{equation}
We should point out
that the result in \cite{beran1985bootstrap}
applies to covariance matrices.
To adapt their results to the Hessian,
we need a central limit theorem
for the estimated Hessian.
Such a result is provided by
Theorem 3 of 
\cite{duong2008feature}
which shows that
\begin{equation}\label{eq::clt}
\sqrt{n}|H|^{1/4}
{\rm vech}
\left[ H^{1/2} \left( \hat\Hess_{Y,h}- \Hess_h\right) H^{1/2}\right]\rightsquigarrow
N(0,\Sigma_2)
\end{equation}
where
$\Sigma_2 = R ({\rm vech} \nabla^{(2)} K)p_h(x)$,
$R(g) = \int g(x) g^T(x) dx$.

Computing
$w^{-1}({\cal S}_j)$ exactly would require
calculating the inverse map $w^{-1}$ explicitly.
We do not know of any computationally efficient method for
computing the inverse map $w$.
However, we do know, by construction, that
$\lambda^{*b} = w^{-1}(s^{*b})$ for each bootstrap sample.
The set
$\bigl\{w^{-1}(s^{*b}):\ s^{*b}\in {\cal S}_j\bigr\}$
approximates
$w^{-1}({\cal S}_j)$ 
arbitrarily well as $B\to\infty$.
Thus, we can approximate 
$w^{-1}({\cal S}_j)$ 
by ${\cal C}_j$ in step 6.

\begin{figure}
\fbox{\parbox{6.5in}{
\vspace{.3cm}
\begin{center}
\sf Local Mode Testing Algorithm
\end{center}

\begin{enum}
\item Split the data into two halves $X$ and $Y$.
\item Using $X$, construct $\hat p_{X,h}$ and use the mean-shift algorithm to find the modes
$\hat{\cal M} = \{\hat m_1,\ldots \hat m_k\}$ of $\hat p_{X,h}$.
\item Using $Y$, find $\hat p_{Y,h}$
and its gradient $\hat g_{Y,h}$ and Hessian $\hat\Hess_{Y,h}$.
\item For each candidate mode $m=\hat m_j\in \hat\Modes$:
\begin{enum}
\item
Estimate the Hessian at $m$ and compute its eigenvalues 
$\hat\lambda_i(m)$,
$i=1, \ldots d.$ 
\item Compute elementary symmetric polynomials $s_i(m)$, 
according to \eqref{eqn::one-to-one} for $i=1,\ldots,d$.
\item
Generate $B$ bootstrap Hessians $\Hess^*_b(m)$, $b=1, \cdots, B$ and obtain
$B$ corresponding 
elementary symmetrical polynomial  $s_i^{*\,b}(m), i=1, \ldots d$.
\item
Compute confidence rectangle $G_j$ and
the confidence set ${\cal C}_j$.
\item Declare that there is a mode at $m=\hat m_j$ if the confidence set ${\cal C}_j$ lies entirely above 0.
\end{enum}
\end{enum}
}}
\caption{\em The Local Mode Testing Algorithm.}
\label{fig::ModeTest}
\end{figure}

We automatically get
confidence intervals for 
the $s^{\rm th}$ negative eigenvalue at the $j^{\rm th}$ mode
$\gamma_{sj}$ where $s=1,\ldots, d$
and $j=1,\ldots, k$.
Thus, in addition to the significance of the mode,
we get valuable information about the shape of the mode,
which we call the {\em eigenportrait.}
This will be illustrated in Section \ref{sec::examples}.

\section{Persistence}
\label{section::persistence}

There is a completely different approach
for eliminating non-significant modes
based on the theory of persistent homology
which has been the focus of recent research
(\cite{chazal2011persistence,
edelsbrunner2008persistent}).
We will not review 
persistent homology here but rather we describe
the salient points that are germane to the present paper.
The key ideas are from
\cite{chazal2011persistence}.

Consider a smooth density $p$ with
$M = \sup_x p(x) < \infty$.
The $t$-level set clusters
are the connected components of the set
$L_t=\{x:\ p(x) \geq t\}$.
Suppose we find the upper level sets
$L_t=\{x:\ p(x) \geq t\}$
as we vary $t$ from $M$ to $0$.
Persistent homology measures how the topology
of $L_t$ varies as we decrease $t$.
In our case,
we are only interested in the modes, which correspond to
the zeroth order homology.
(Higher order homology refers to holes, tunnels etc.)

Imagine setting $t=M$
and then gradually decreasing $t$.
Whenever we hit a mode, a new level set cluster is born.
As we decrease $t$ further, some clusters may merge
and we say that one of the clusters (the one born most recently) has died.
See Figure \ref{fig::explain}.

In summary, each mode $m_j$ has a death time and a birth time
denoted by $(d_j,b_j)$.
(Note that the birth time is larger than the death time because
we start at high density and move to lower density.)
The modes can be summarized with a persistence diagram
where we plot
the points
$(d_1,b_1),\ldots, (d_k,b_k)$ in the plane.
See Figure \ref{fig::explain}.
Points near the diagonal correspond to modes with short lifetimes.
\cite{chazal2011persistence}
suggested killing any mode with a short lifetime
$\ell_j = b_j - d_j$.
This requires choosing a significance threshold.
\cite{balakrishnan2013statistical}
suggest that this threshold can be based on
the bootstrap quantile $\epsilon_\alpha$ defined by
\begin{equation}
\epsilon_\alpha = \inf \Biggl\{ z:\ 
\frac{1}{B}\sum_{b=1}^B I\Bigl( ||\hat p_h^{*b} - \hat p_h||_\infty > z\Bigr) \leq \alpha \Biggr\}.
\end{equation}
Here,
$\hat p_h^{*b}$ is the density estimator based on the
$b^{\rm th}$ bootstrap sample.
This corresponds to killing a mode if it is in a
$2\epsilon_\alpha$ band around the diagonal.

\begin{figure}
\begin{center}
\includegraphics[scale=.5]{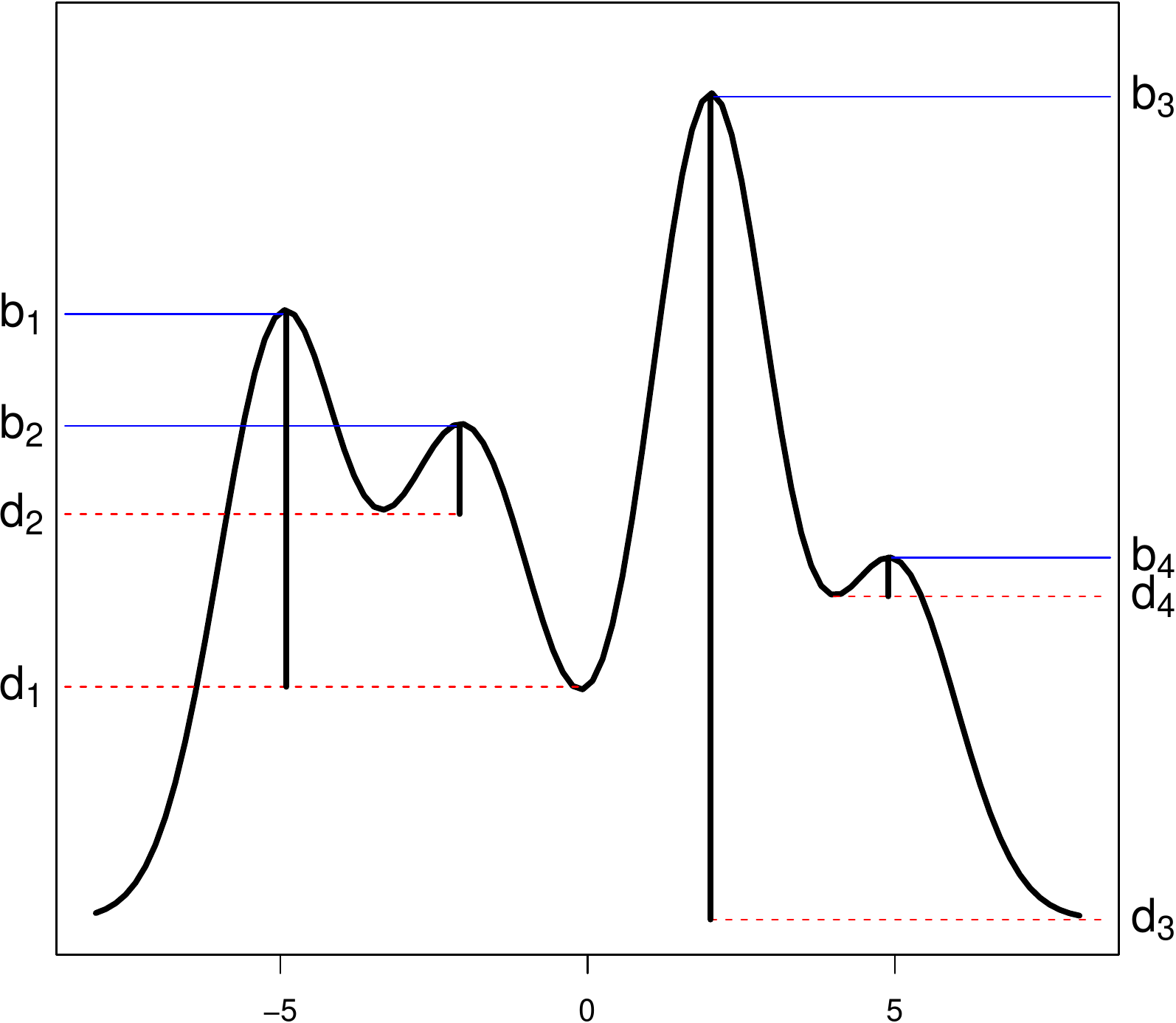}\includegraphics[scale=.5]{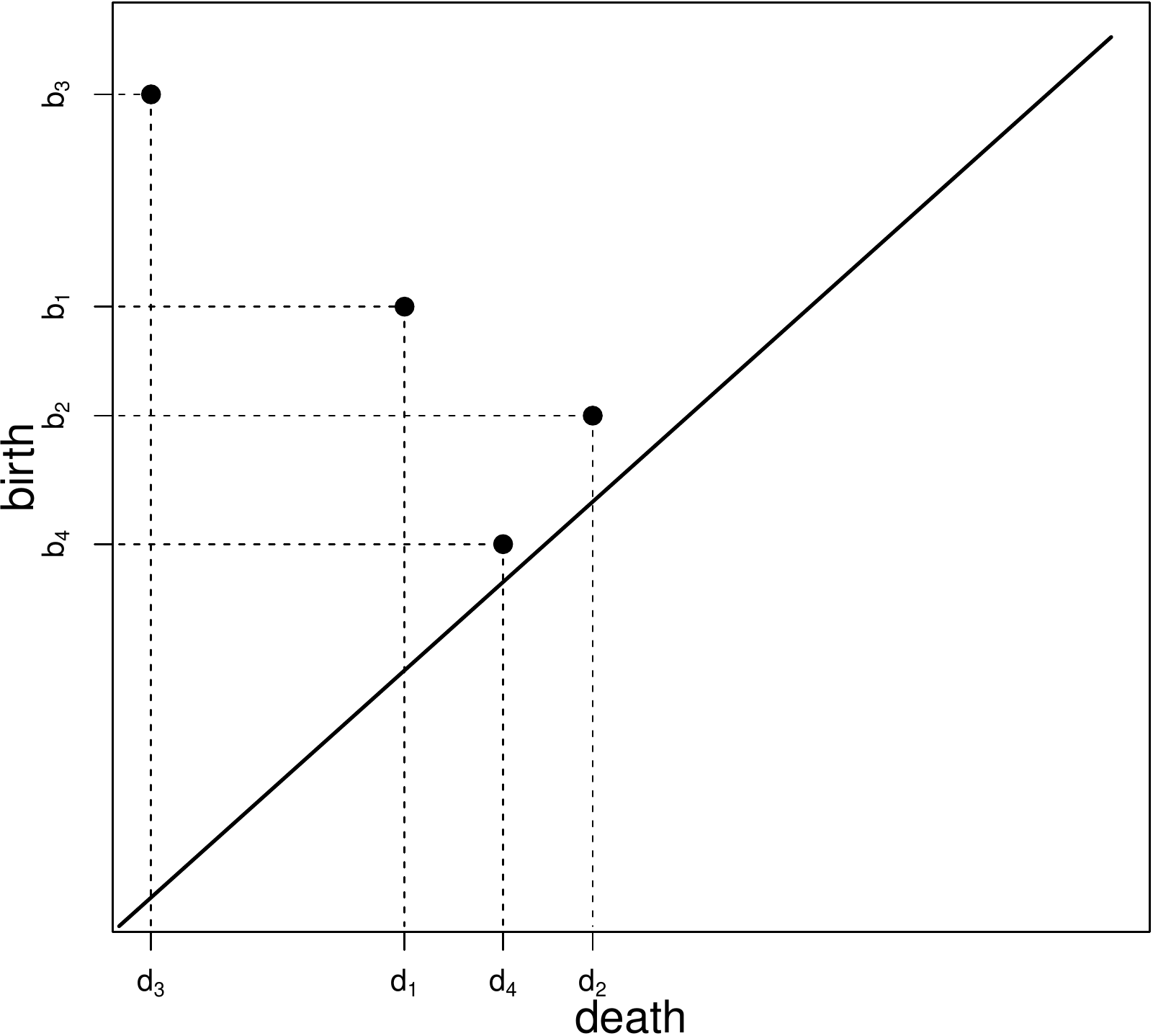}
\end{center}
\caption{\em Starting at the top of the density and moving down,
each mode has a birth time $b$ and a death time $d$.
The persistence diagram (right) plots the points
$(d_1,b_1),\ldots, (d_4,b_4)$.
Modes with a long lifetime are far from the diagonal.}
\label{fig::explain}
\end{figure}

The local test method proposed in this paper
and the persistence method, each have advantages and disadvantages.
The advantages of the persistence approach are that
it does not require data-splitting, it does not require
estimating derivatives
and that, when used in its complete form, it can be used to
find higher-order topological features.
Also, the persistence diagram provides a simple
visualization, independent of the dimension of the data.

The advantages of the local method
are that it provides more shape information about each mode
(via the confidence intervals for the eigenvalues of the Hessian)
and that it is much faster since the bootstrap
is only computed at $k$ points.
In comparison,
the bootstrap for the persistence approach has to be computed
over a fine grid
to approximate 
$||\hat p_h^{*j} - \hat p_h||_\infty$.
Also, the local method never needs to compute the
persistence of the modes which is itself computationally expensive.

In summary, there are advantages and
disadvantages to each approach
and in fact, they both provide
useful information.
The methods are less
similar when one considers
higher-order structure.
The natural extension of local modes
to higher-dimensional objects
corresponds to ridges and hyper-ridges as in
\cite{GPVW}.
In contrast, high-order persistent homology corresponds to
holes and tunnels.
Thus, the two approaches are aimed at different types of structure.

One thing that persistence and local eigenportraits have in common
is that both permit visualization of data regardless
of the dimension of the data.

\section{Examples}
\label{sec::examples}

We start with a few simple examples
to illustrate the method.
Figure \ref{fig::simple-examples}
shows four one-dimensional examples.
In each case $n=200$ and $\alpha = 0.10$.
The first column shows kernel density estimators
and the second column shows confidence intervals for $\gamma_1$
at each mode.

The first two rows
are based on data from a Normal distribution.
Row 1 has a bandwidth of $h=1$
and we find one significant mode.
In row 2 we use a small bandwidth namely
$h=.1$.
In this case there are numerous potential modes
but each is declared to be non-significant
as is evident from
the plot of the confidence intervals for the
$\gamma_{1j}$'s.
This shows an important feature of our procedure:
false modes that occur by using a bandwidth that is too small
are correctly regarded as random fluctuations rather
than being significant modes.
This can be used as a diagnostic to alert us that the bandwidth is too small.
We discuss this point further in Section \ref{section::bandwidth}.
The next two rows show the results
for a mixture of two Normals
($n=200$, $h=1$ and $\alpha=.10$)
and a mixture of three Normals
($n=200$, $h=1.5$ and $\alpha=.10$).
The method correctly finds the appropriate modes.

\begin{figure}
\begin{center}
\includegraphics[scale=.6]{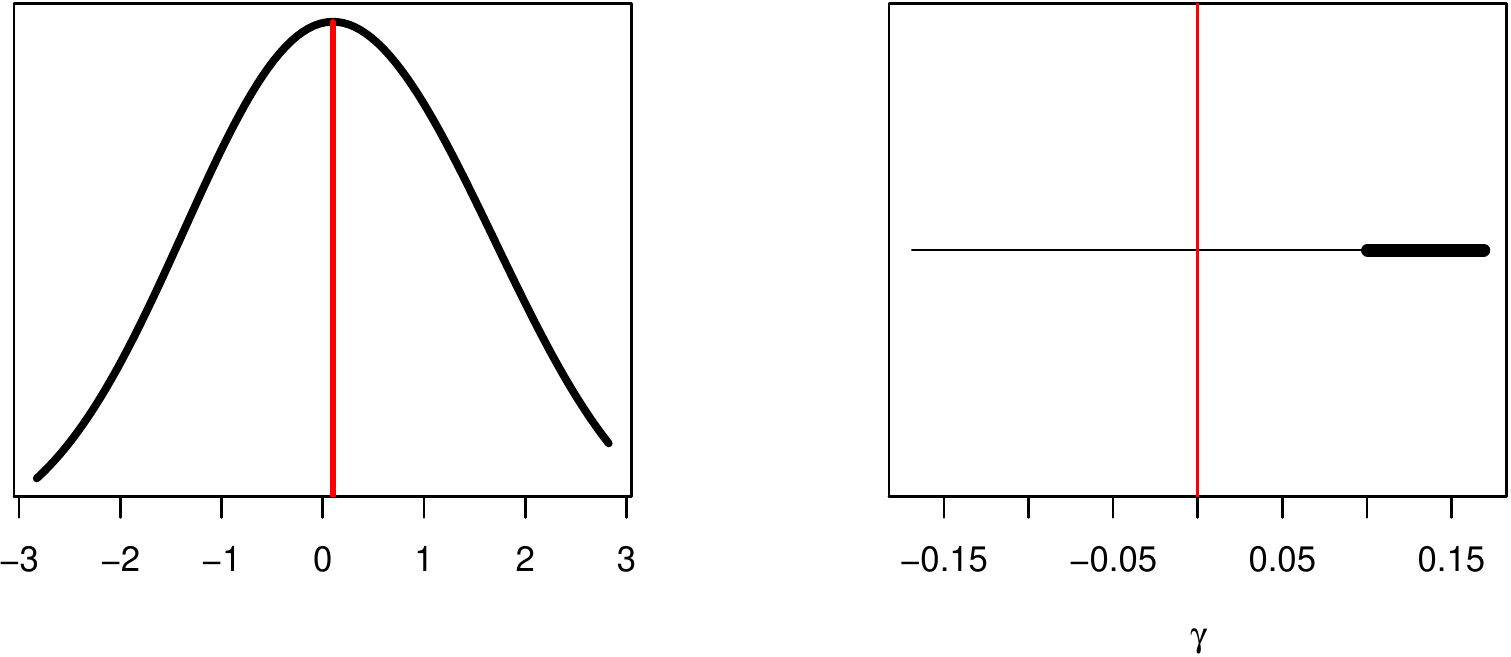}
\includegraphics[scale=.6]{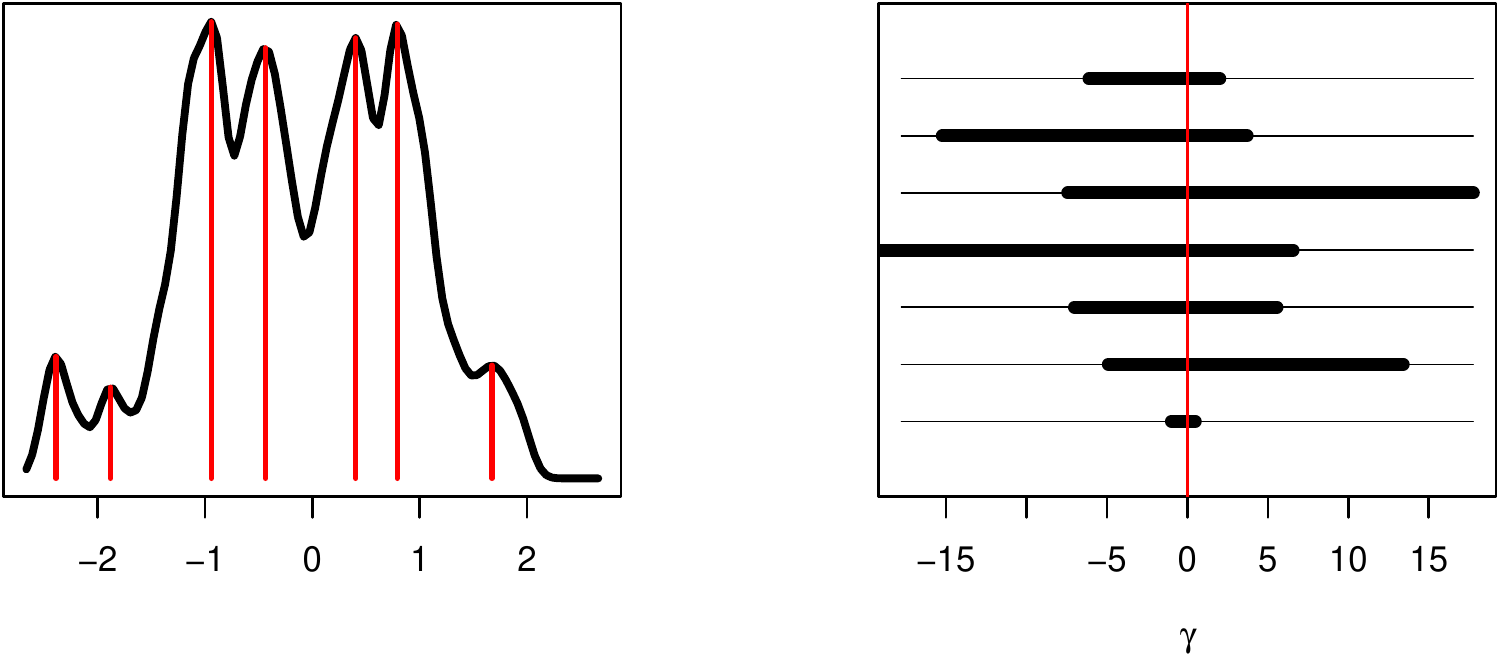}
\includegraphics[scale=.6]{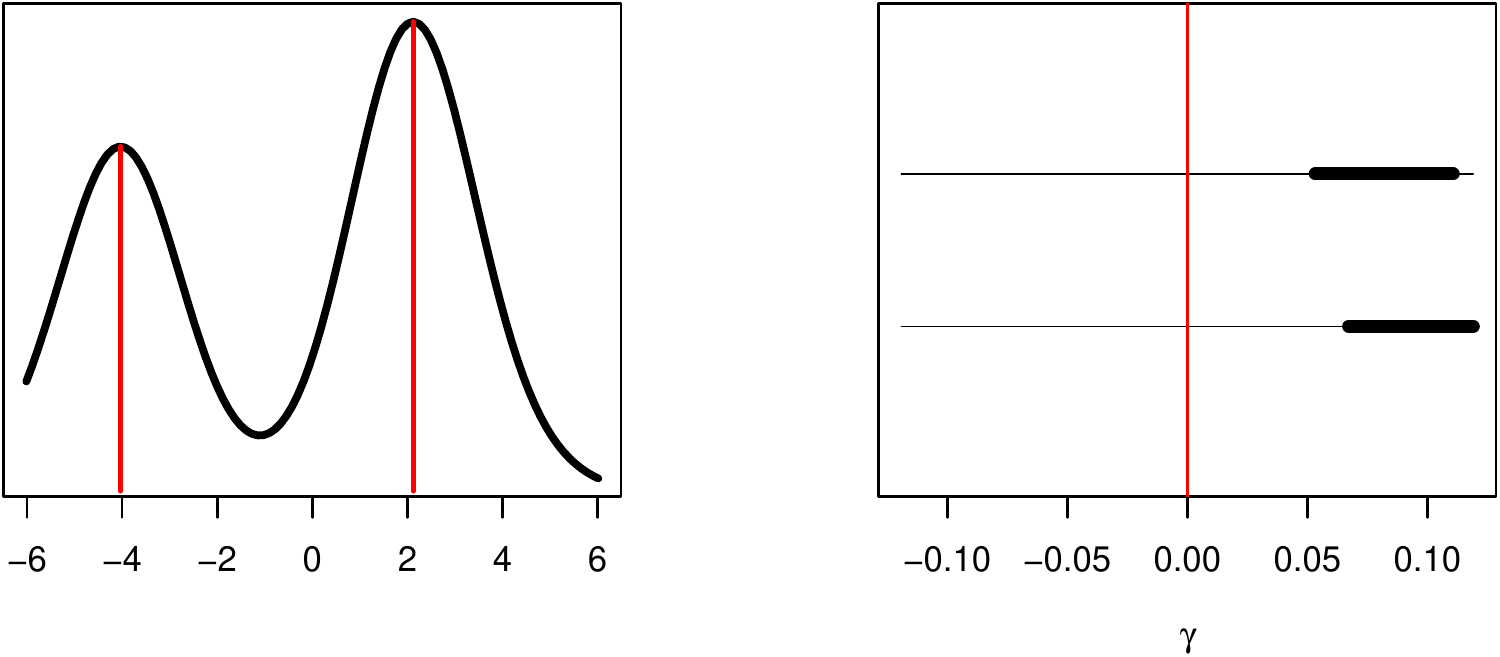}
\includegraphics[scale=.6]{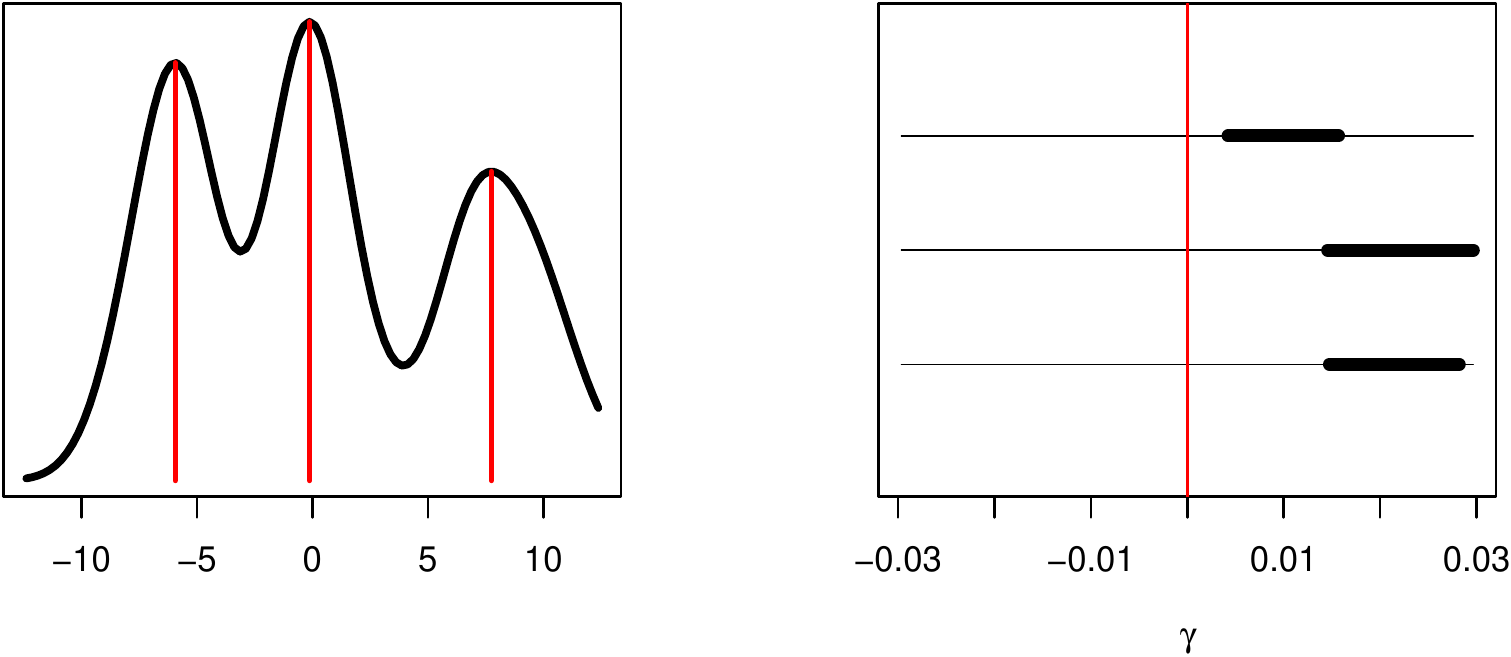}
\end{center}
\caption{\em 
Left plots are kernel density estimators.
Right plots are confidence intervals for the $\gamma_{1j}$'s.
Top row: Normal data, $n=200$, $h=1$, $\alpha = 0.10$.
Second row: Normal data, $n=200$, $h=0.1$, $\alpha = 0.10$.
Third row: a mixture of two Normals, $n=200$, $h=1$ and $\alpha=.10$.
Fourth row: a mixture of three Normals, $n=200$, $h=1$ and $\alpha=.10$.}
\label{fig::simple-examples}
\end{figure}

The confidence intervals in Figure
\ref{fig::10D2M} are for a 10-dimensional dataset with two modes
(a mixture of two Gaussians).
The true density is
$p(x) = \frac{1}{2} \phi(x; \mu_1,\Sigma_1) + \frac{1}{2} \phi(x; \mu_2,\Sigma_2)$
where
$$
\mu_1 = (-5,\ldots, -5),\ \ \ \mu_2 = (5,\ldots, 5),
$$
$\Sigma_1$ is the identity matrix
and
$\Sigma_2$ is diagonal with
diagonal entries
$(1,1,1,1,1,.01,.01,.01,.01,.01)$.
We used $n=10,000$,
$h=1$ and $\alpha = 0.05$.
The procedure located four modes.
The plots in Figure \ref{fig::10D2M}
are done per mode, rather than per eigenvalue.
That is, there is one plot for each of the four modes,
each showing the eigenportrait of all ten eigenvalues.
We see that only two of the modes are significant.
Thus two modes are correctly labeled as not real.
The eigenportraits of the two
significant modes are interesting.
The first eigenportrait shows the spherical nature of the mode.
The second shows that the mode is very non-spherical.
Thus we have an informative way to visualize the 10-dimensional data.

Figure \ref{fig::4Modes-h3} shows a two-dimensional example
with four modes that have different shapes.
The eigenportrait reveals that one mode is highly
non-spherical.

The data in Figure \ref{fig::ring}
show what happens when our
assumptions are violated.
The data have three well-separated modes.
There is also a ring which, technically,
consists of infinitely many, non-separated modes.
Both the persistence method and the local testing method
declare the three spherical modes to be significant.
The modes on the ring are declared non-significant by both methods.
The eigenportrait nicely distinguishes the
difference in shape of the different modes.

\begin{figure}
\begin{center}
\includegraphics[scale=.7]{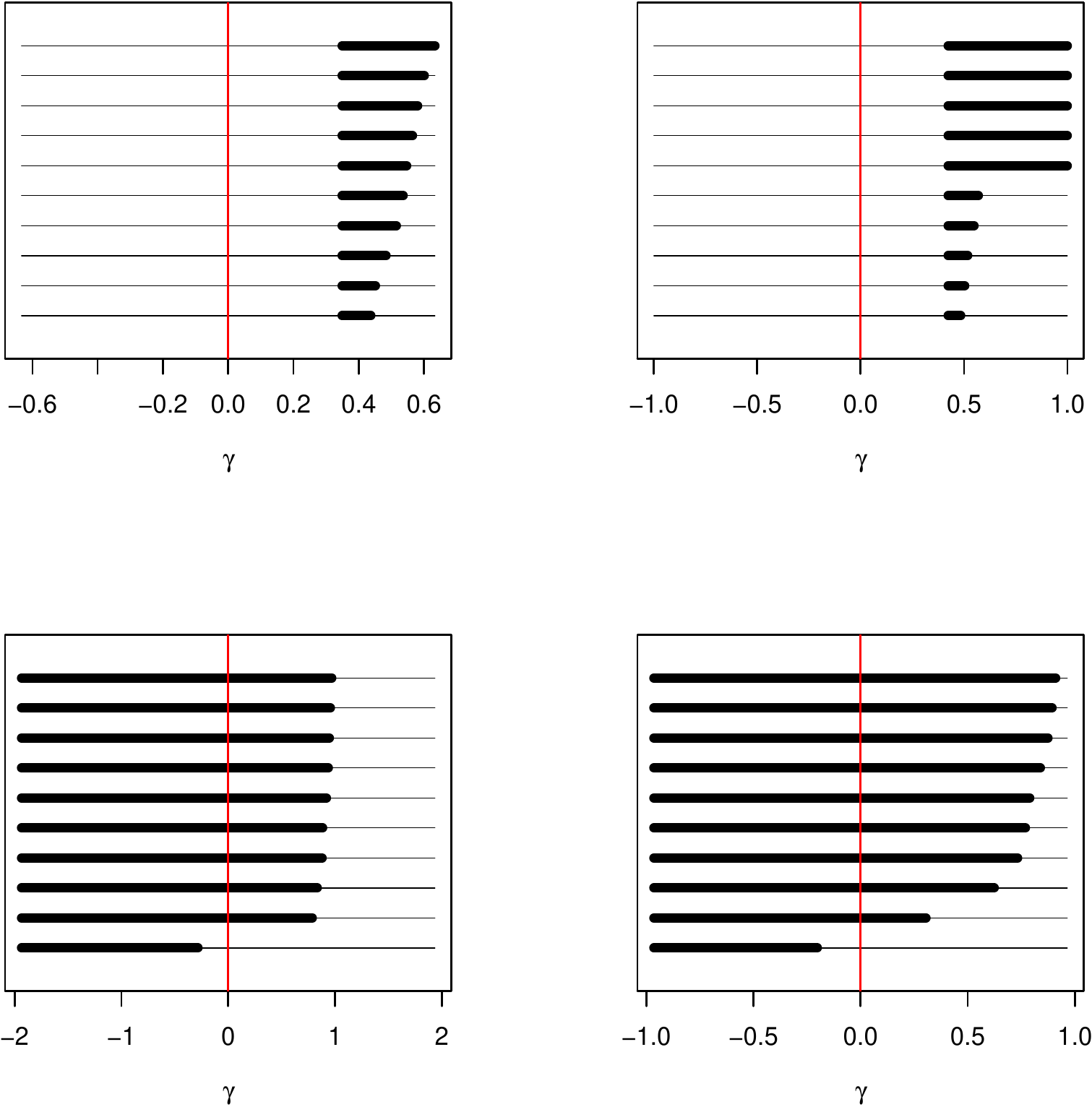}
\end{center}
\caption{\em Eigenportrait of 10-dimensional data. 
Each plot shows confidence intervals for all 10 $\gamma_j$'s.
The top two plots show the two significant modes.
The bottom two plots show the two non-significant modes.
Note that the eigenportraits of the significant modes
show that the two modes have different shapes.}
\label{fig::10D2M}
\end{figure}

\begin{figure}
\includegraphics[scale=1]{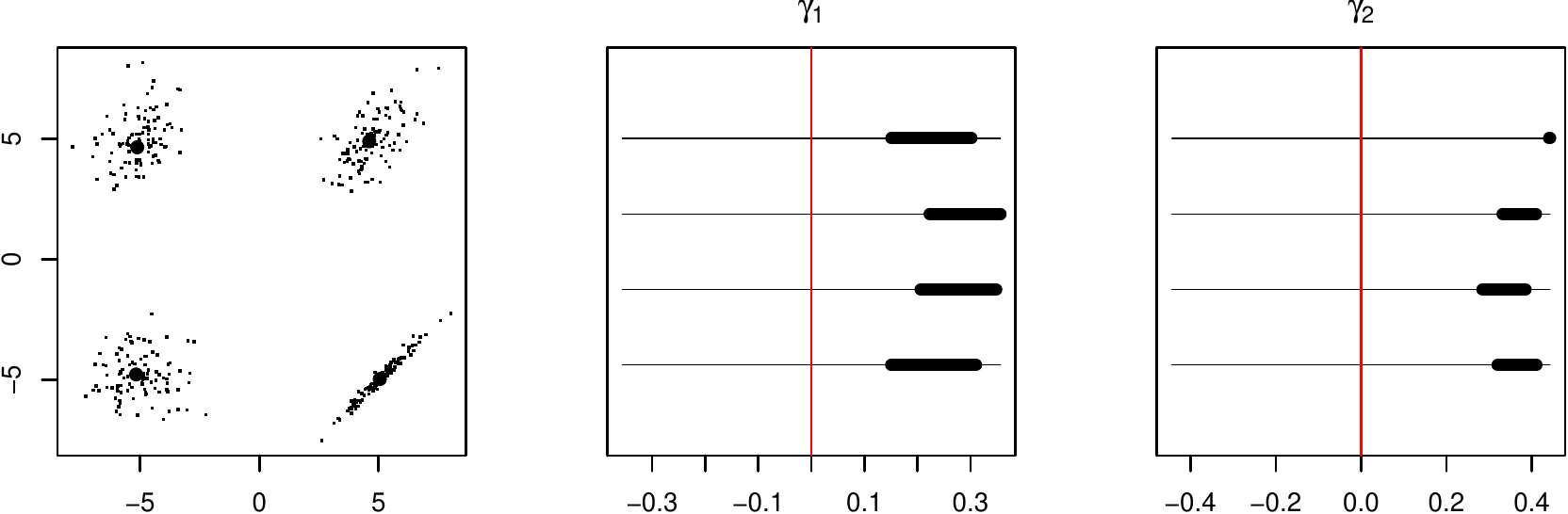}
\caption{\em Left: Data with 4 modes of differing shapes.
Middle: confidence intervals for $\gamma_{1j}$ at each mode.
Right: confidence intervals for $\gamma_{2j}$ at each mode.}
\label{fig::4Modes-h3}
\end{figure}

\begin{figure}
\begin{center}
\includegraphics[scale=1]{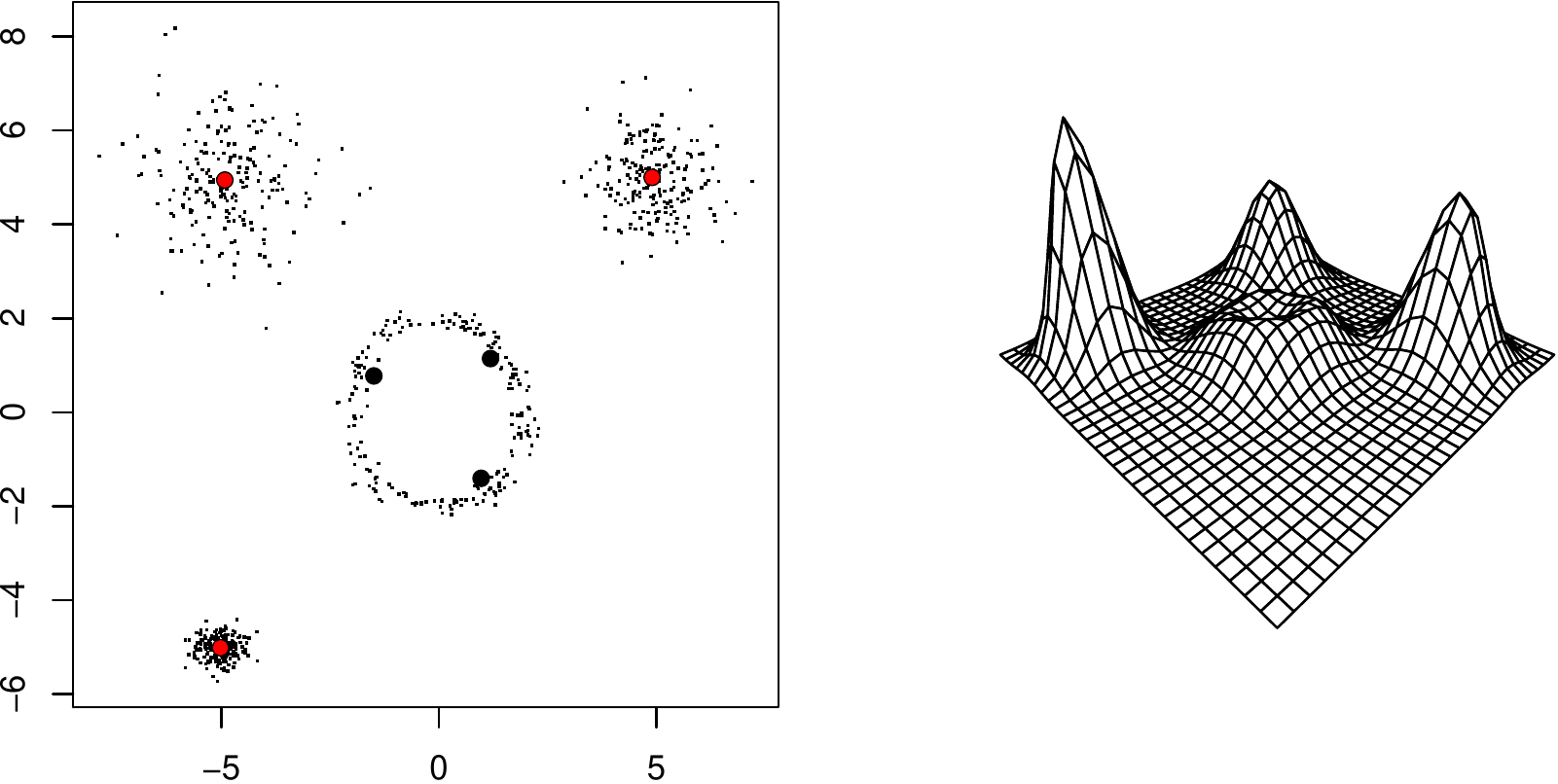}\\
\includegraphics[scale=1]{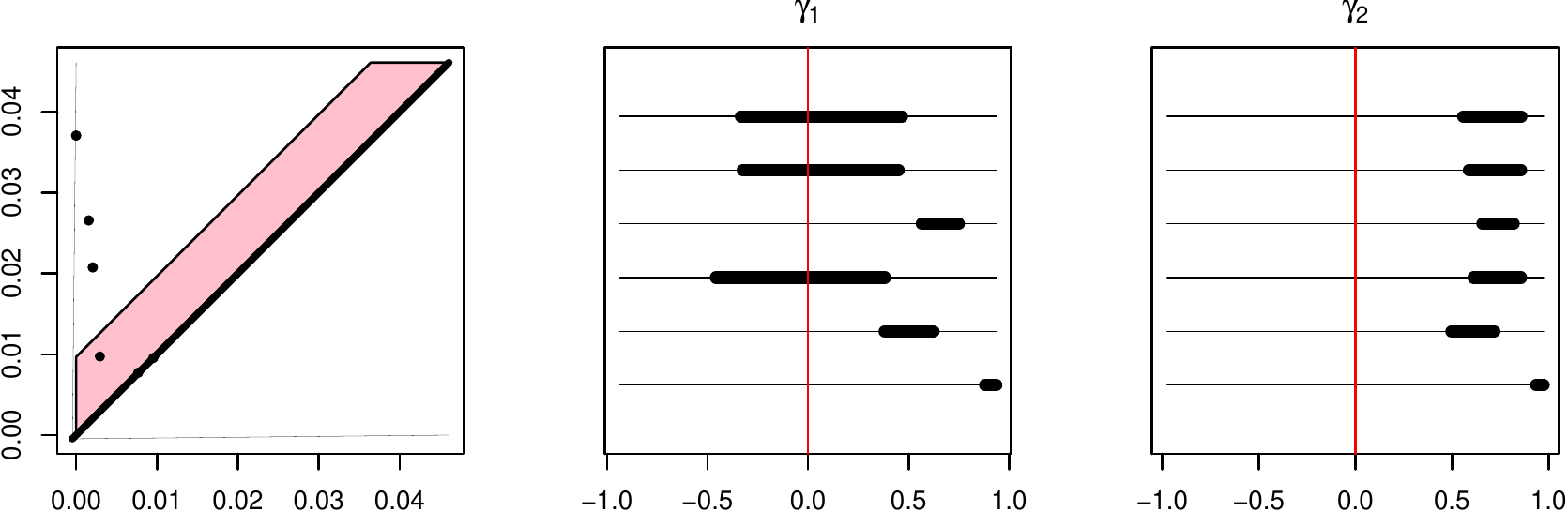}
\end{center}
\caption{\em Top: 3 modes plus a ring.
Bottom left: persistence diagram showing three significant modes
and several non-significant modes.
Points in the filled-in strip are declared to be non-significant.
Middle: confidence intervals for $\gamma_1$ at each mode. (Three significant modes.)
Right: confidence intervals for $\gamma_2$ at each mode.}
\label{fig::ring}
\end{figure}

Now we turn to the
earthquake data 
analyzed in
\cite{duong2008feature}.
The data are the epicenters of 512 earthquakes
before the 1982 eruption of Mt St Helens.
The data, and the three significant modes we found are shown in Figure
\ref{fig::earthquakedata}.
The three variables are
latitude, longitude and
$-\log(- {\rm depth})$.
We use a bandwidth of .3
(which is roughly consistent with the analysis in
\cite{duong2008feature}.)
Figure \ref{fig::earthquake}
shows the persistence diagram
and the eigenportrait.
All the analyses are consistent with three modes
and three different depths.
The modes we located are consistent with the regions of interest
found in 
\cite{duong2008feature}.
The eigenportraits show that $\gamma_1$
(corresponding to depth)
has the most uncertainty (larger confidence intervals).
This makes sense since the latitude
and longitude have small variation.

\begin{figure}
\begin{center}
\includegraphics[scale=.7]{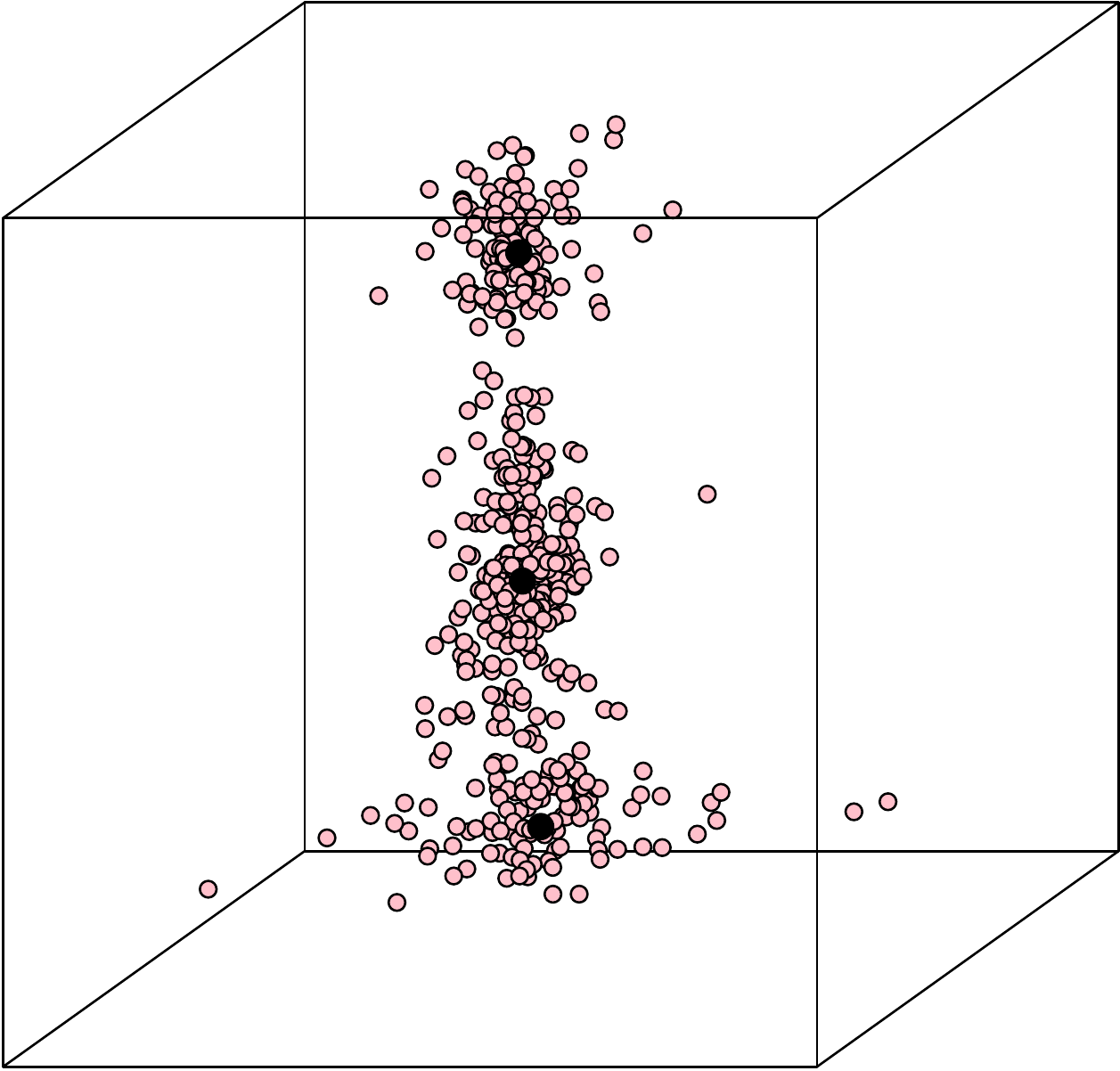}
\end{center}
\caption{\em Scatterplot of the earthquake data.
The three dark points are the estimated modes.}
\label{fig::earthquakedata}
\end{figure}

\begin{figure}
\begin{center}
\includegraphics[scale=.5]{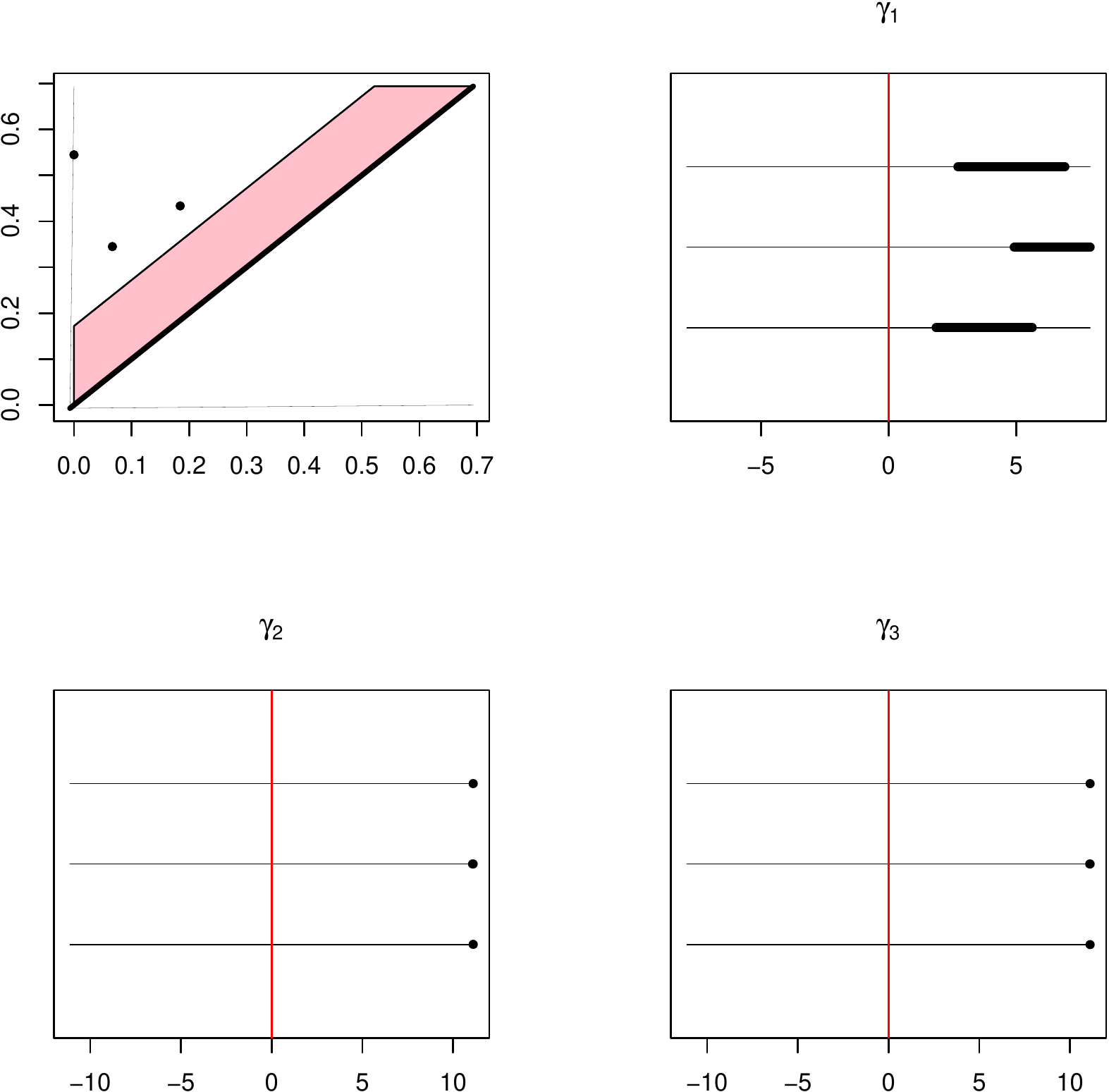}
\end{center}
\caption{\em The earthquake data. Top left: the persistence diagram shows three
significant modes.
Points in the filled-in strip are declared to be non-significant.
Top right: confidence intervals for $\gamma_1$.
Bottom left: confidence intervals for $\gamma_2$.
Bottom right: confidence intervals for $\gamma_3$.}
\label{fig::earthquake}
\end{figure}

\section{A Possible Method For Choosing the Bandwidth}
\label{section::bandwidth}

Bandwidth selection for mode hunting
is a challenging problem.
The first method 
we are aware of is \cite{silverman1981using}
although it has not been used much in practice.
Recent, very promising work has focused on
accurate estimation of derivatives;
see \cite{ChaconAndDuong, chaconMonfort}.
The results in this paper
suggest another approach to selecting the bandwidth
for mode hunting.
The purpose of this section is to
briefly (and heuristically) introduce the idea.

\begin{figure}
\begin{center}
\includegraphics[scale=.85]{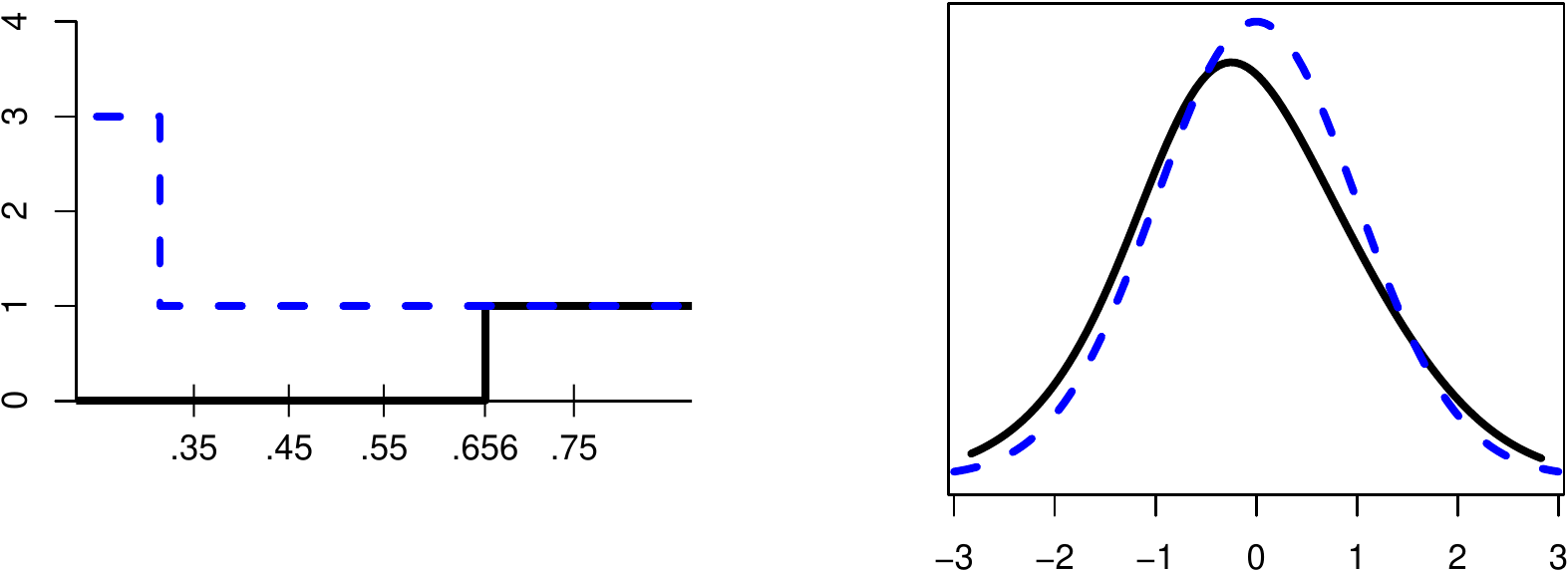}\\
\vspace{.2in}
\includegraphics[scale=.85]{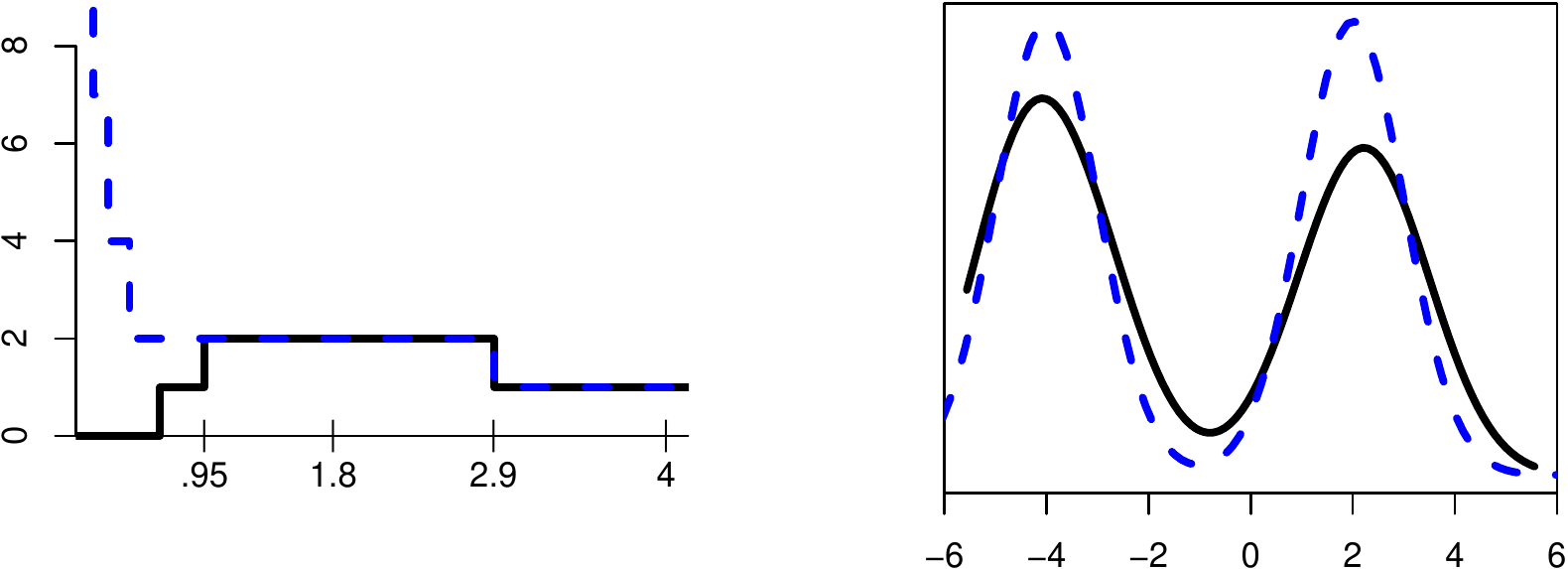}\\
\vspace{.2in}
\includegraphics[scale=.85]{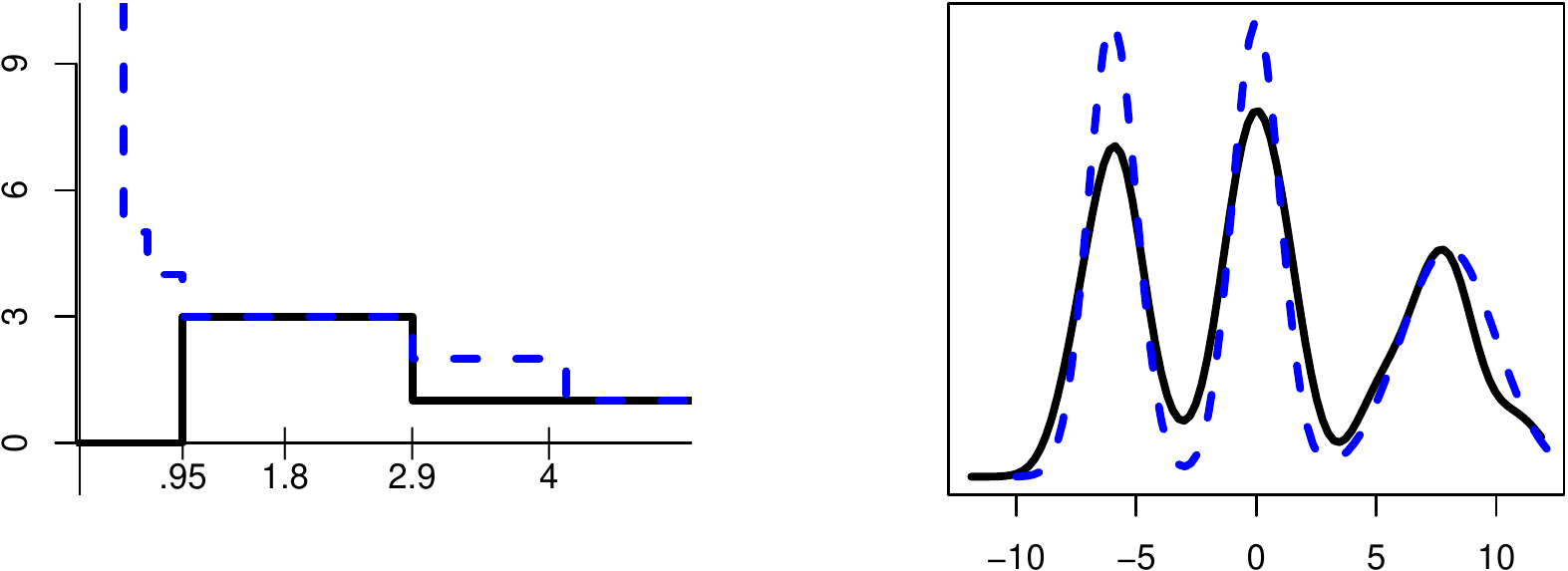}
\end{center}
\caption{\em Left:
Number of modes (dashed line) and number of significant modes (solid line) modes versus bandwidth.
Right: true density (dashed) and estimated density (solid)
using bandwidth $\hat h$ given in equation (\ref{eq::choose-h}).}
\label{fig::BDW}
\end{figure}

We have seen that when the bandwidth $h$ is chosen to be small,
many modes are found but our procedure
identifies these modes
as random fluctuations
in the estimated density.
On the other hand, when $h$ is large,
there will be at most one mode.

Thus, while the number of modes decreases with $h$,
the number of significant modes
is small when $h$ is either too small or too large.
If mode finding is our main goal, 
rather than accurate estimation in the $L_2$ norm,
then this suggests
a new way to choose the bandwidth $h$:
choose $h$ to maximize
the number of significant modes.
More precisely,
let $N(h)$ be the number of significant modes
found by our test, as a function of $h$.
Let $m = \max\{ N(h):\ h> 0\}$
and define
\begin{equation}\label{eq::choose-h}
\hat h = \inf \Bigl\{h:\ N(h) = m\Bigr\}.
\end{equation}

We now examine the result of applying this procedure
in a few examples.
Figure
\ref{fig::BDW}
shows the number of modes
and the number of significant modes $N(h)$ versus bandwidth
for a Normal (top), a mixture of  two Normals (middle) and
a mixture of three Normals (bottom).
In each case, choosing the bandwidth to
maximize the number of significant modes
leads to the correct number of modes.

Now we turn to a very challenging problem:
selecting a bandwidth when the density is singular.
Consider, for example, the distribution,
$$
P=\frac{1}{3} N(-\mu,\sigma) + \frac{1}{3}\delta_0 + 
\frac{1}{3} N(\mu,\sigma)
$$
where $\delta_0$ is a point mass at 0.
Of course, $P$ does not even have a density.
Nonetheless, $p_h$ has three modes
and a kernel density estimator will indeed show three modes
for certain values of $h$.
If we apply the usual cross-validation method,
we will get $\hat h=0$ because there are ties in the data.
This leads to a useless estimator.
What can we hope for in this example?
Estimating well in the $L_2$ sense does not even make sense.
Instead, we at least hope to get a density estimator
with three modes.
Figure \ref{fig::singular} shows an example with
$\mu =10$, $\sigma=1$ and $n=180$.
Here we see that we do indeed get three modes.

\begin{figure}
\begin{center}
\includegraphics[scale=1.0]{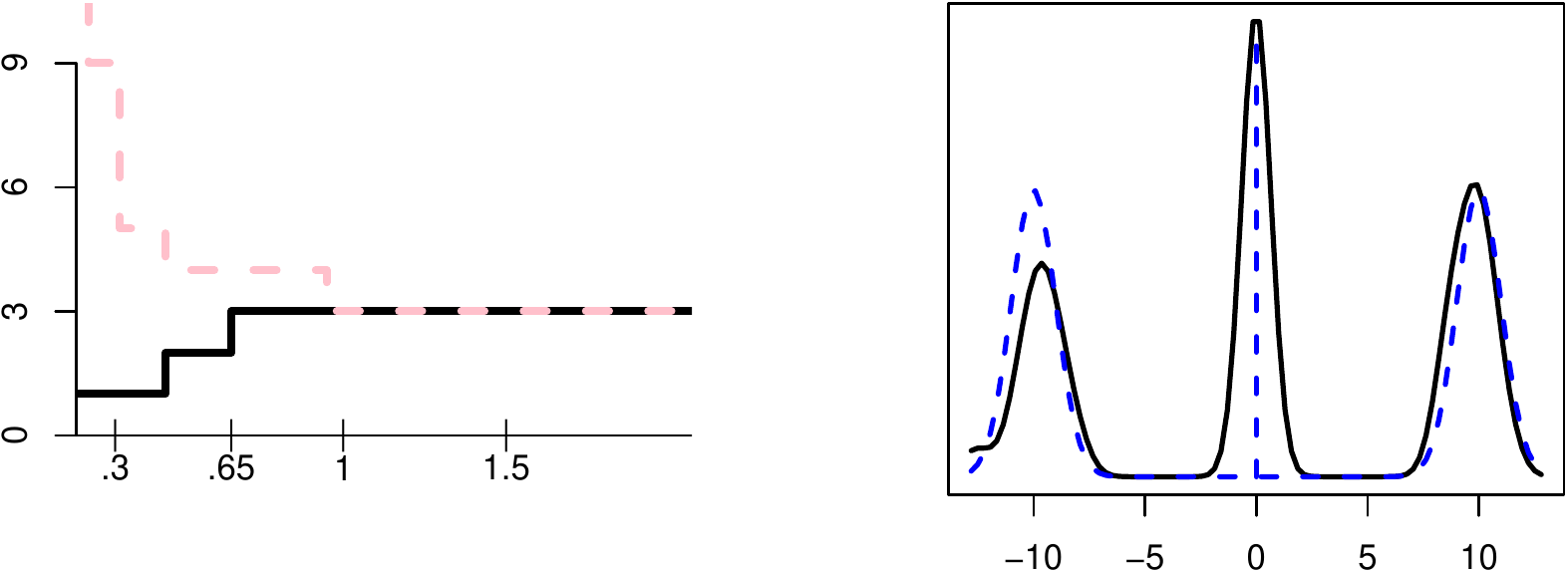}
\end{center}
\caption{\em The true distribution is a mixture of two Normals and a 
point mass at 0.
Left: number of modes (dashed) and number of significant modes (solid) versus bandwidth.
Right: density estimator using $\hat h$.
The estimate is solid.
The true density is dashed.
We use a vertical line to indicate the singular component.
The proposed bandwidth selection method chooses a bandwidth leading
to a smooth density with three modes.}
\label{fig::singular}
\end{figure}

These results are very encouraging but, of course,
a thorough investigation of the idea is needed before
it can be recommended for general use.
To establish theoretical properties
of this method requires theory that is valid when $h\to 0$.
Unfortunately, the usual asymptotic theory requires that
$h^{d+4} n \to \infty$ which precludes small bandwidths.
Hence, a rigorous theory for this method remains an open problem.

\section{Theoretical Properties}
\label{section::theory}

We examine here some theoretical properties of the procedure 
described in Sections \ref{section::testing} and \ref{section::esp}.
Our main goal is to bound the width of the confidence interval for $\gamma_1$
(Theorem \ref{theorem::width}.)
A secondary goal is to show that the modes discovered in stage 1 of the procedure
are good estimators of the true modes.
This fact has been established in various papers
but we could not find an explicit statement of the result
in the multivariate, multi-mode case so we include the details for completeness.
We begin by restating the assumptions.

(A1) The density $p$ is a bounded, continuous
density supported on a compact set
${\cal X}\subset \mathbb{R}^d$.

(A2) The gradient $g$ and Hessian ${\cal H}$ of $p$
are bounded and continuous.
We assume that the Hessian is non-degenerate at every stationary point.

(A3) $p$ has finitely many 
modes $m_1,\ldots, m_{k_0}$
in the interior of ${\cal X}$.

(A4)
Let
\begin{equation}
\Delta = \min_{j\neq k}||m_j - m_k|| \ \ \ \ \ \mbox{and}\ \ \ \ \ 
L=\max_{1\leq j\leq k_0}\lambda_1({\cal H}(m_j)).
\end{equation}
We assume that 
$\Delta >0$ and $L < 0$.

(A5) The kernel $K$ is a symmetric probability density with bounded and continuous first
and second derivatives and bounded second moment.

{\bf Properties of $p_h$.}
Recall that $\hat p_h$
is the kernel density estimator with
bandwidth matrix $H=h^2 I$.
Let
$\hat g_h$
and $\hat {\cal H}_h$
are the gradient and Hessian of $\hat p_h$.
Let
\begin{equation}
p_h(x) = \int K(t) p(x+ th) dt
\end{equation}
be the mean of the kernel density estimator.
Let $g_h$ and ${\cal H}_h$ denote the gradient and Hessian of $p_h$.
For $h>0$ small enough,
$p_h$ inherits all the above properties.
The proofs of the following two lemmas
are standard and are omitted.

\begin{lemma}
Assume (A1)-(A5).
Assume that $h^2 < C \epsilon$ for some $C$.
Then, for all $h>0$ and $\epsilon$ small enough we have:
\begin{enum}
\item $p_h$ is a bounded and continuous density.
\item The gradient $g_h$ and Hessian ${\cal H}_h$ of $p_h$
are bounded and continuous.
\item $p_h$ has finitely many 
modes $m_{1h},\ldots, m_{k_0 h}$
in the interior of ${\cal X}$
where $\max_j ||m_j - m_{jh}|| = O(h^2)$.
\item $\Delta_h >0$ and $L_h < 0$ 
where
$$
\Delta_h = \min_{j\neq k}||m_{jh} - m_{kh}|| \ \ \ \ \ \mbox{and}\ \ \ \ \ 
L_h=\max_{1\leq j\leq k_0}\lambda_1({\cal H}_h(m_{jh})).
$$
\end{enum}
\end{lemma}

The conditions also guarantee that $p$ and $p_h$
are locally quadratic around their modes.
Let $B(x,\epsilon)$ denote a ball of radius $\epsilon$
centered at $x$.

\begin{lemma}\label{lemma::quadratic}
Assume that $h^2 < C \epsilon$ for some $C$.
Let $B_j = B(m_j,\epsilon)$.
When $\epsilon$ and $h$ are small enough,
$m_{jh}\in B_j$ for each $j$.
Moreover, 
there exists $\delta>0$ and $c>0$ such that
the following are true:
$$
\max_j \sup_{x\in B_j}\lambda_1({\cal H}(x))    \leq -\delta
\ \ \ \mbox{and}\ \ \ 
\max_j \sup_{x\in B_j}\lambda_1({\cal H}_h(x))  \leq -\delta
$$
and, for all $j$ and all $x\in B_j$,
$$
p(m_j) - p(x) \geq c ||x-m_j||^2\ \ \mbox{and}\ \ \ 
p_h(m_{jh}) - p_h(x) \geq c ||x-m_{jh}||^2.
$$
\end{lemma}

{\bf Properties of $\hat p_h$.}
Here we record some useful facts about $\hat p_h$.
We have that
\begin{align}
\sup_{x\in{\cal X}} ||\hat p_h(x) - p_h(x)|| &\leq C \sqrt{\frac{\log n}{nh^d}}\nonumber \\
\sup_{x\in{\cal X}} ||\hat g_h(x) - g_h(x)|| &\leq C \sqrt{\frac{\log n}{nh^{d+2}}} \label{eq::bounds}\\
\sup_{x\in{\cal X}} ||\hat {\cal H}_h(x) - {\cal H}_h(x)|| &\leq C \sqrt{\frac{\log n}{nh^{d+4}}}\nonumber
\end{align}
almost surely, for all large $n$.
The first bound is proved in
\cite{gine2002rates}
and the bounds on $\hat g_h$ and $\hat{\cal H}_h$ follow similarly.
From Theorems 1 and 3 of
Duong, Cowling, Koch and Wand (2008),
we have that
\begin{equation}
\sqrt{n h^{d+2}}(\hat g_h(x) - g_h(x))\rightsquigarrow N(0,\Sigma_1)
\end{equation}
where
$\Sigma_1 = p_h(x) \int \nabla K(x) \nabla K(x)^T dx$
and
\begin{equation}
\sqrt{n h^{d+4}} {\rm vech} (\hat{\cal H}_h(x) - {\cal H}_h(x)) \rightsquigarrow N(0,\Sigma_2)
\end{equation}
where
$\Sigma_2 = p_h(x) \int ({\rm vech}\nabla^{(2)} K(x) )({\rm vech}\nabla^{(2)} K(x) )^T dx.$

{\bf Properties of the Estimated Modes.}
Let 
\begin{align*}
{\cal M}     &= \mbox{modes of} \ p\\
{\cal M}_h   &= \mbox{modes of} \ p_h\\
\hat{\cal M} &= \mbox{modes of} \ \hat{p}_{X,h}\\
{\cal M}^\dagger &= \mbox{modes in\ }\hat{\cal M}\ \mbox{that pass the hypothesis test in Stage 2}.
\end{align*}

\begin{lemma}\label{lemma::mode-hat}
Let $\epsilon>0$ and $h>0$ be sufficiently small with
$h^2 < C \epsilon$ for some $C>0$.
Let $B_j = B(m_{j},\epsilon)$.
Let $B_0 = {\cal X} - \bigcup_{j=1}^k B_j$.
Then, as $n\to \infty$:
\begin{enum}
\item $\mathbb{P}(B_j \bigcap \hat{\cal M}\neq \emptyset\ {\rm for\ all\ }j)\to 1$.
Thus, $\hat p_{X,h}$ has at least one mode in each $B_j$.
\item With probability tending to 1,
$\hat p_{X,h}$ has exactly one mode $\hat m_{jh}$ in each $B_j$.
\item Let $x_n$ be any maximizer of $\hat p_{X,h}$ in $B_j$.
Then
\begin{equation}
||x_n - m_{jh}|| = O_P\left(  (nh^d)^{-1/4}\right)
\end{equation}
and
\begin{equation}
||x_n - m_{j}|| = O_P\left(  (nh^d)^{-1/4}\right) + O(h^2).
\end{equation}
\end{enum}
\end{lemma}

\begin{proof}
(1) 
Since $\hat p_{X,h}$ is a bounded continuous function,
it has a maximizer over $B_j$.
We claim that the maximizer must be in the interior of $B_j$.
Write
$B_j = A_0\cup A_1 \cup A_2$ where
$A_0 = \{x:\ ||x-m_j|| \leq \epsilon/3\}$,
$A_1 = \{x:\ \epsilon/3 < ||x-m_j|| \leq 2\epsilon/3\}$ and
$A_2 = \{x:\ 2\epsilon/3 < ||x-m_j|| \leq \epsilon\}$.
For $C$ large enough,
$m_{jh}\in A_0$.
Also, from the properties of $p_h$,
$$
\inf_{x\in A_0} p_h(x) > \sup_{x\in A_2} p_h(x).
$$
With probability tending to 1,
$$
\inf_{x\in A_0}\hat p_{X,h}(x) \geq
\inf_{x\in A_0} p_h(x) - C \sqrt{\frac{\log n}{n h^d}}>
\sup_{x\in A_2} p_h(x) - C \sqrt{\frac{\log n}{n h^d}}
\geq \sup_{x\in A_2} \hat p_{X,h}(x) + C \sqrt{\frac{\log n}{n h^d}}.
$$
So,
with probability tending to 1,
$$
\inf_{x\in A_0}\hat p_{X,h}(x) > \sup_{x\in A_2}\hat p_{X,h}(x).
$$
This implies that any maximizer $x$ of $\hat p_{X,h}$ over $B_j$
is in the interior of $B_j$
and hence
$\hat g_{X,h}(x) = (0,\ldots, 0)^T$.
Furthermore,
$$
\lambda_1(\hat{\cal H}_{X,h}(x)) \leq \lambda_1({\cal H}_{X,h}(x)) +o_P(1) \leq -\delta + o_P(1).
$$
So, with probability tending to 1,
$\hat p_{X,h}$ has a maximizer $x$ in the interior of $B_j$ with 0 gradient
and negative Hessian eigenvalues, i.e. it is a mode.

(2) Suppose $\hat p_{X,h}$
has two modes $x$ and $y$ in $B_j$.
So
$\hat g_{X,h}(x)=\hat g_{X,h}(y)=(0,\ldots,0)^T$.
Recall the exact Taylor expansion
for a vector valued function $f$ is
$f(a+t) - f(a) = t^T \int_0^1 Df(a + ut) du.$
So
$$
(0,\ldots, 0)^T =
\hat g_{X,h}(y) - \hat g_{X,h}(x) = 
(y-x)^T \int_0^1 \hat{\cal H}_{X,h}(x + u(y-x)) du
$$
and hence, multiplying on the right by by $y-x$,
\begin{align*}
0 &=  \int_0^1 (y-x)^T \hat{\cal H}_{X,h}(x + u(y-x)) (y-x) du \leq
||y-x|| \sup_u \lambda_1(\hat{\cal H}_{X,h}( x + u(y-x)))\\
& \leq
||y-x|| \sup_u [\lambda_1({\cal H}_{X,h}( x + u(y-x))) + o_P(1)]\\
& \leq
||y-x||  \Bigl[-\delta + o_P(1)\Bigr] < - \frac{\delta ||y-x||}{2}
\end{align*}
with probability tending to one.
Hence, $x=y$.

(3) This proof uses a strategy similar to that in Donoho and Liu (1991, Theorem 5.5).
Let $x$ be any maximizer of $\hat p_{X,h}$ over $B_j$.
Then
$\hat p_{X,h}(x) \geq \hat p_{X,h}(m_{jh})$ 
(where $m_{jh}\in {\cal M}_h$)
and hence
$$
[\hat p_{X,h}(x) - p_h(x)] - [\hat p_{X,h}(m_{jh}) - p_h(m_{jh})] \geq
p_h(m_{jh}) - p_h(x) \geq c||m_{jh}- x||^2
$$
where we used
Lemma \ref{lemma::quadratic}.
Hence,
$$
Z_n(x) - Z_n(m_{jh}) \geq 
c \sqrt{nh^d}  ||m_{jh}- x||^2
$$
where
$Z_n(x) = \sqrt{nh^d}(\hat p_{X,h}(x) - p_h(x))$.
It can be shown that
$\sup_{x\in {\cal X}}||Z_n(x)|| = O_P(1)$.
Hence,
$$
||m_{jh}- x||^2 \leq \frac{1}{c \sqrt{nh^d} }\sup_x ||Z_n(x)|| = O_P\left(\sqrt{\frac{1}{nh^d} }\right).
$$
Hence,
\begin{equation}
||x - m_{jh}|| = O_P\left(\frac{1}{nh^d} \right)^{1/4}.
\end{equation}
\end{proof}

{\bf Properties of the ESP Transformation.}
By construction, the $1-\alpha$ asymptotic confidence set ${\cal S}$ 
in \eqref{eq::setS} is a $d$-dimensional hypercube in $\Re^d$. 
The confidence interval ${\cal C}$ for $\lambda_1$ is
$$
R=R({\cal S}) = \Bigl[\inf\,\bigl\{a\in Q\bigr\},\ \ \sup\,\bigl\{a\in Q\bigr\} \Bigr]
$$
where
$Q=w^{-1}({\cal S})$.
In this section,
we bound the size of $R$.

Let
${\cal B}$ be the set of all symmetric
$d\times d$ matrices and let
${\cal E} = \bigl\{w(\lambda(A)):\ A\in {\cal B}\bigr\}$.
Thus, if
$s=(s_1,\ldots, s_d)\in {\cal E}$
then $w^{-1}(s)$ corresponds to
the eigenvalues of some symmetric matrix.
Let ${\cal S}={\cal S}(s_0,\epsilon)$ be any hyper-cube in $\mathbb{R}^d$:
$$
{\cal S}(s_0,\epsilon) = \Bigl\{ t\in\mathbb{R}^d:\ ||s-t||_\infty \leq \epsilon\Bigr\}
$$
for some $s_0$ and $\epsilon$.
We want to bound the size of
$$
R=R({\cal S}) = \Bigl[\inf\,\bigl\{a\in Q\bigr\},\ \ \sup\,\bigl\{a\in Q\bigr\} \Bigr]
$$
where
$Q=w^{-1}({\cal S}\bigcap {\cal E})$.

Each 
$s \in {\cal S}\cap {\cal E}$ defines a characteristic polynomial
\begin{equation}\label{eqn::polyn}
P_s(\lambda) = \prod_{i=1}^d (\lambda_i-\lambda) =
\lambda^d + \sum_{k=1}^d (-1)^k \ s_k  \ \lambda^{d-k} = 0
\end{equation}
whose roots are the eigenvalues of some symmetric matrix.

\begin{lemma} \label{lem::interv}
There exists $C>0$, 
depending only $\epsilon_0$ and $s_0$, such that,
for all $\epsilon< \epsilon_0$,
$$
C\epsilon + o(\epsilon) \leq \mu(R({\cal S}(s_0,\epsilon)\cap {\cal E}  )) \leq C \epsilon^{1/d}.
$$
\end{lemma}

\begin{proof}
Without loss of generality, assume that $d$ is even.
(A simple modification of the proof works for $d$ odd.)
First, note that, 
there is some $L>0$ (depending on $s_0$ and $\epsilon_0$)
such that for all $\epsilon < \epsilon_0$
and all $s\in {\cal S}(s,\epsilon)$,
we have
$-L \leq \lambda_d(s) \leq \lambda_1(s) \leq L$.
Let
$s,\tilde{s}\in {\cal S}\bigcap {\cal E}$ so
that
$||\tilde s- s_0||_\infty \leq 2\sqrt{d}\epsilon$.
Let $P_{s}$ and $\tilde P_{\tilde s}$ be the polynomials corresponding to 
$s$ and $\tilde s$. Then
\begin{equation}\label{eqn::diff}
|P_s(\lambda) - \tilde{P}_{\tilde{s}}\ (\lambda)| \leq 
      \sum_{k=1}^d |(-1)^k|\ |s_k - \tilde{s}_k|\ |\lambda^{d-k}| \leq
     \ 2\sqrt{d} \ \epsilon \sum_{k=1}^d  |\lambda|^{d-k}\ \leq \ C \ \epsilon
\end{equation}
where
$C = 2\sqrt{d} \sum_{k=1}^d L^{d-k}$.
Let  $\lambda$ and $\tilde{\lambda}$ be the ordered eigenvalues of $P_s$ and 
$\tilde{P}_{\tilde s}$. 
First, suppose 
$\tilde{\lambda}_1 > \lambda_1$. For 
all $\lambda > \lambda_1$, the polynomial in \eqref{eqn::polyn}
can be written as $P_s(\lambda) = \prod_{i=1}^d (\lambda - \lambda_i)$
showing that it is an increasing function of $\lambda$, since each 
factor in the product is increasing. Let 
$\lambda_1 < t < \tilde{\lambda}_1$, 
then
$$
P_s\,(\tilde\lambda_1) = \prod_{i=1}^d (\tilde\lambda_1 - \lambda_i) \geq (\tilde\lambda_1 - \lambda_1)^d.
$$
From \eqref{eqn::diff}
$$
C\ \epsilon \ \geq |P_s(\tilde{\lambda}_1) - \tilde{P}_{\tilde s}\ (\tilde{\lambda}_1)| 
\ = \ |P_s(\tilde{\lambda}_1)|\ \geq\ \ (\tilde{\lambda}_1-\lambda_1)^d.
$$
Hence,
$\lambda_1 \leq \tilde\lambda_1 \leq (C\epsilon)^{1/d}$.
Now assume $\tilde{\lambda}_1 < \lambda_1$. 
Then
$$
\tilde{P}_{\tilde s}\ (\lambda_1) = 
\prod_{i=1}^d (\lambda_1 - \tilde{\lambda}_i) \geq (\lambda_1- \tilde{\lambda}_1)^d .
$$
Similarly, from \eqref{eqn::diff} 
$$
C\ \epsilon \ \geq |P_s(\lambda_1) - \tilde{P}_{\tilde s}(\lambda_1)| 
           \ = \ |\tilde{P}_{\tilde s}(\lambda_1)|\ \geq\  |\tilde{\lambda}_1 -\lambda_1|^d.
$$
Thus $|\tilde{\lambda}_1 -\lambda_1|< C\ \epsilon^{1/d}$. 
The lower bound follows by choosing some point 
$s \in {\cal S}(s_0,\epsilon)\cap {\cal E}$
that is in the interior of ${\cal E}$.
For such a point, $\lambda$ is a continuously differentiable function of $s$
and the bound follows from a simple Taylor expansion.
\end{proof}

{\bf Remark:}
The worst case is when
$\lambda_1 = \cdots = \lambda_d$ and
the characteristic polynomial is simply
$(\lambda_1-\lambda)^d$.
In that case, a small perturbation of $s$ can cause
a perturbation of $\lambda_1$ of size $O(\epsilon^{1/d})$.

\vspace{.1in}

{\bf Properties of the Confidence Interval and Test.}

\begin{theorem}\label{theorem::width}
Let ${\cal C}_j$ be the confidence interval for
$\gamma_{1j}=-\lambda_1({\cal H}_h(x))$
for any $x$.
Then the Lebesgue measure is
$$
\mu({\cal C}_j) = O_P\left( (n h^{d+4})^{-1/d}\right).
$$
\end{theorem}

{\bf Proof Outline.}
We can write
$s = f({\cal H}_h)$
for some smooth, continuously differentiable function $f$.
The asymptotic variance
of $\hat{\cal H}_{Y,h}$ is of order
$\epsilon_n$ where
$\epsilon_n =O_P( (n h^{d+4})^{1/2})$.
It may then be shown that the $1-\alpha$ confidence rectangle for $s$
has size of order $\epsilon_n$.
The result then follows 
From Lemma \ref{lem::interv},
the size of of ${\cal C}_j$ is
$O(\epsilon_n^{1/d})$. $\Box$

\begin{lemma}
Let $B_j = B(m_j,\epsilon)$.
We have:
\begin{enum}
\item $\mathbb{P}(B_j \cap {\cal M}^\dagger\neq \emptyset\ {\rm for\ all\ }j)\to 1$.
\item Let $B_0 = \{x:\ \lambda_1({\cal H}_h(x)) \geq 0\}$.
Then $\limsup_{n\to\infty}\mathbb{P}(\hat{\cal M}^\dagger\cap B_0 \neq\emptyset)\leq \alpha$.
\end{enum}
\end{lemma}

\begin{proof}
(1) In parts (1) and (2) of Lemma \ref{lemma::mode-hat}
we showed there exists one mode
$\hat m_{jh}\in B_j$ with zero gradient and
negative eigenvalues.
In Theorem \ref{theorem::width},
we showed that the width of the confidence interval for the first
eigenvalue of the Hessian
at $\hat m_{jh}$ shrinks to 0.
This implies that, with probability tending to 1,
the test rejects the null and hence
$\hat m_{jh}$ is included in ${\cal M}^\dagger$.

(2) Let $x\in B_0$.
Then $x\in {\cal M}^\dagger$ if and only if
$x\in\hat{\cal M}$ and if
the confidence interval excludes the true value of
$\lambda_1({\cal H}_h(x))$.
Let $U = \hat{\cal M}\cap B_0$.
Conditional on $X$,
the probability that the test rejects the null for any $z\in Z$
has, asymptotically, probability at most $\alpha/k$.
Hence,
$\mathbb{P}(U \neq \emptyset|X) \leq \alpha + o(1)$ and,
by the independence of $X$ and $Y$,
$\mathbb{P}(U \neq \emptyset) \leq \alpha + o(1)$.
\end{proof}

{\bf When the Bandwidth is Small.}
When $h$ is small, we get spurious modes
which are killed off by the hypothesis test.
This behavior is clear in the examples.
Intuitively, it follows since the size of confidence rectangle
increases as $h$ decreases.
We have seen numerically that
this prevents us from choosing
a bandwidth that is too small
because the number of significant modes
becomes 0 when $h$ is too small.
Making this fact rigorous remains an open question.
When $h$ gets very small, the usual
asymptotic methods no longer apply.
It is possible that
uniform-in-bandwidth asymptotics
(\cite{einmahl2005uniform})
might be useful here
but this is beyond the scope of the paper
and we leave this to future work.

\section{Discussion}
\label{section::discussion}

We have introduced a new method for testing the significance
of modes in density estimators.
There are several ideas that
we hope to deal with in future work.
These include the following:

\begin{enumerate}
\item Our method complements the approaches
in \cite{duong2008feature} and
\cite{chazal2011persistence}
by providing extra information
about the estimated modes.
A thorough investigation
into combining the strengths
of all three methods deserves
future work.
\item If one makes specific assumptions about the size and separation of the modes,
then it should be possible to find the asymptotic power of the test.
\item We indicated a possible method for choosing the bandwidth
for mode hunting.
Deriving precise theoretical properties
of the method 
will require techniques that allow small bandwidths.
\item 
The ultimate goal of this line of work is to show that
the clusters based on the significant modes
are a good approximation to the
population clusters
${\cal A}_1,\ldots, {\cal A}_{k_0}$
defined in
(\ref{eq::ascending}).
The results in this paper
are only a first step towards that goal.
We would like to show, in fact, that
with high probability,
${\cal A}_j$ contains one and only one significant mode.
Furthermore,
\cite{chacon} suggest an interesting risk function
for mode clustering.
We conjecture that deleting non-significant modes
before clustering may improve the risk of mode-based clustering.
Also, we conjecture that our bandwidth selection method
will lead to good clustering risk.
\end{enumerate}

\bibliography{paper}

\end{document}